\documentclass[reprint,prl,superscriptaddress]{revtex4-1}
\usepackage[dvips]{graphicx}
\usepackage{amsmath,amssymb,amsthm,mathrsfs,amsfonts}
\usepackage{textcomp}
\usepackage{physics}
\usepackage{bm}
\usepackage{enumerate}
\usepackage{soul,xcolor}
\usepackage{verbatim}
\usepackage{color}
\usepackage{hyperref}

\newtheorem{theorem}{Theorem}
\newtheorem{lemma}{Lemma}
\newtheorem{corollary}{Corollary}

\newtheorem{remark}{Remark}
\newtheorem{result}{Result}

\newcommand{\renyi}{R$\mathrm{\acute{e}}$nyi }
\newcommand{\mc}{\mathcal}

\newcommand{\N}{\mathcal{N}}

\begin{document}

\setstcolor{red}   


\title{Operational Resource Theory of Quantum Channels}

\begin{abstract}

Quantum resource theories have been widely studied to systematically characterize the non-classicality of quantum systems. 
Most resource theories focus on quantum states and study their interconversions.
Although quantum channels are generally used as the tool for state manipulation,  such a manipulation capability can be naturally regarded as a generalized quantum resource, leading to an open research direction in the resource theories of quantum channels.
Various resource-theoretic properties of channels have been investigated, however, without treating channels themselves as operational resources that can also be manipulated and converted. 
In this Letter, we address this problem by first proposing a general resource framework for quantum channels and introducing resource monotones based on general distance quantifiers of channels. 
We study the interplay between channel and state resource theories by relating resource monotones of a quantum channel to its manipulation power of the state resource.
Regarding channels as operational resources, we introduce asymptotic channel distillation and dilution, the most important tasks in an operational resource theory, and show how to bound the conversion rates with channel resource monotones.
Finally, we apply our results to quantum coherence as an example and introduce the coherence of channels, which characterizes the coherence generation ability of channels. We consider asymptotic channel distillation and dilution with maximally incoherent operations and find the theory asymptotically irreversible, in contrast to the asymptotic reversibility of the coherence of states. 


\end{abstract}

\date{\today}

\author{Yunchao Liu}
\affiliation{Department of Electrical Engineering and Computer Sciences, University of California, Berkeley, California 94720, USA}
\affiliation{Center for Quantum Information, Institute for Interdisciplinary Information Sciences, Tsinghua University, Beijing 100084, China}

\author{Xiao Yuan}
\email{xiao.yuan.ph@gmail.com}
\affiliation{Department of Materials, University of Oxford, Parks Road, Oxford OX1 3PH, United Kingdom}

\maketitle

Quantum resource theories have been developed as systematic frameworks for the characterization, quantification, and operational interpretation for various quantum effects, including coherence~\cite{aberg2006quantifying,Baumgratz14,RevModPhys.89.041003}, discord~\cite{henderson2001classical,PhysRevLett.88.017901,Modi12}, entanglement~\cite{PhysRevA.53.2046,PhysRevLett.78.2275,RevModPhys.81.865},  thermodynamics~\cite{PhysRevLett.111.250404,goold2016role}, magic in stabilizer computation~\cite{veitch2012negative,veitch2014resource,PhysRevLett.118.090501}, etc. The advances in quantum resource theories not only lead to a deeper understanding of the underlying physics, but also provide new insights and mathematical tools for various quantum information processing tasks that exploit the resources, such as quantum key distribution~\cite{devetak2005distillation,coles2016numerical}, quantum random number generation~\cite{YuanPhysRevA2015,yuan2016interplay,PhysRevA.97.012302}, and quantum computing~\cite{Datta09,veitch2014resource,PhysRevX.6.021043, howard2014contextuality,PhysRevA.93.012111,anand2016coherence}. We refer to Ref.~\cite{chitambar2018quantum} for a recent review.

A quantum resource theory usually starts by defining three important components: free states, free operations and resource measures. Free states are those quantum states that do not possess any resource. Free operations are quantum operations that cannot generate resource from free states, and their precise definitions are guided by physical motivations. Resource measures are functionals that map quantum states to real numbers, which cannot be increased under free operations. 
In an operational resource theory, one of the most important tasks is to study state conversion under free operations.
Resource distillation and dilution are  optimal schemes that convert between a given state and the maximal resource state, which in general can be uniquely determined in a given resource theory.
In many resource theories, such as coherence and entanglement, the distillation and dilution tasks are generally characterized by resource measures based on the relative entropy and $\alpha$-\renyi divergences 
in the asymptotic i.i.d.~\cite{rains2001semidefinite,hayden2001asymptotic,winter2016operational} and the general one-shot scenario~\cite{Buscemi2010distill,Brandao2011oneshot,Buscemi2011entanglement,zhao2018oneshot,regula2018one,zhao2018oneshotdistill,liu2019one}, respectively. 


Existing developments in quantum resource theories are mainly centered around quantum states, while quantum channels are used as the tool for state resource manipulation. 
In principle, we can also regard a quantum channel as the resource object and study the resource theory of channels.
This has been done mainly for characterizing a certain property of quantum channels, such as channel simulation~\cite{Berta2013entanglement,pirandola2017fundamental,wilde2018entanglement,Dana2018resource,Diaz2018usingreusing}, spatial correlations~\cite{Rivas2015quantifying}, resource generation~\cite{li2018quantifying,takagi2019general}, magic quantification~\cite{seddon2019quantifying, wang2019quantifying}, entanglement of quantum channels~\cite{BHLS03,Kaur_2017}, channel discrimination~\cite{PhysRevLett.118.100502,berta2018amortized,pirandola2019fundamental}, quantum memory~\cite{Rosset2018resource,2018arXiv180903403S,yuan2019robustness}, non-Gaussianity~\cite{zhuang2018resource}, and others~\cite{quantifying17sr,2018arXiv180502431C,2018arXiv181110307K,wolfe2019bell}. Some general conditions that resource theories of channels should satisfy were also formulated~\cite{coecke2016mathematical,theurer2018quantifying}. Nevertheless, a high-level view of a quantum channel itself as an operational resource has not been well established, partially because the entropic quantifiers for quantum channels has only been developed recently~\cite{Cooney2016,Leditzky2018approaches,yuan2019hypothesis,gour2018entropy}.
It remains an open direction to study the  operational resource theory of quantum channels~\cite{chitambar2018quantum}.


In this Letter, we introduce a general framework for the resource theory of quantum channels. We make use of distance quantifiers of channels, which are defined based on distance quantifiers of states, to construct resource monotones of channels. Because the resource theories of states and channels are highly related, we show how the channel resource monotones bound its ability of manipulating state resources. Then we introduce channel distillation and dilution, two  fundamental tasks in an operational channel resource theory. By focusing on the resource generation capability of channels, we show how the asymptotic distillation and dilution rates are bounded by the introduced resource monotones.
Finally, we take coherence as an example to show the applicability of our general results. We propose the relative entropy of channel coherence and show how it measures the ability of coherence generation. We study channel distillation and dilution with maximally incoherent operations and obtain tight bounds for the asymptotic rates, concluding that the theory is irreversible.

\emph{A general resource framework of channels.---}We first review the resource theory of states. Denote the Hilbert space by $\mc{H}$, quantum states by $\mc{D}(\mc{H})$, and quantum channels by $\mc{L}(\mc{H})$. A \emph{state resource theory} in a Hilbert space $\mc{H}$ is defined as a tuple $(\Omega, \Phi,\mu)$. Here $\Omega\subseteq\mc{D}(\mc{H})$ is the set of free states; $\Phi\subseteq\mc{L}(\mc{H})$ is the set of free or resource-nongenerating (RNG) operations that preserve free states, i.e., $\phi(\omega)\in\Omega, \forall\phi\in\Phi, \forall\omega\in\Omega$; and $\mu:\mc{D}(\mc{H})\to\mathbb{R}$ are resource measures that satisfy:\\
(S1) Nonnegativity: $\mu(\rho)\geq 0$ and $\mu(\omega)=0,\forall\omega\in\Omega$;\\
(S2) Monotonicity: $\mu(\phi(\rho))\leq \mu(\rho),\forall \rho\in\mc{D}(\mc{H}), \forall \phi\in\Phi$.\\
When these two properties are satisfied, we called it a resource monotone. Additional requirements of resource measures can be added for a specific resource theory. 
One popular type of resource measure is based on distance quantifiers $D(\rho,\sigma)$, which satisfy nonnegativity ($D(\rho,\sigma)\geq 0$ and $D(\rho,\sigma)= 0$ if $\rho=\sigma$) and data-processing inequality ($D(\mc{N}(\rho),\mc{N}(\sigma))\leq D(\rho,\sigma),\forall\mc{N}\in \mc L(\mc H)$). 
Resource quantifiers are defined as the minimal distance to the set of free states, $\mu(\rho)=\min_{\omega\in\Omega}D(\rho,\omega)$, which satisfies (S1) and (S2).


Following a similar mathematical structure, a \emph{channel resource theory} is also defined as a tuple $(\mc{F},\mc{O},\mc{R})$. Free channels $\mc{F}$ are those quantum channels that do not have any resource. 
Free super-operations $\mc{O}$ are a subset of super-channels~\cite{Chiribella2008transforming} that transforms free channels to free channels. 
Here super-channels transform a quantum channel $\mc{N}^{A\to B}$ to another channel $\mc{M}^{C\to D}$ by
  $\mc{M}^{C\to D}=\mc{W}^{BE\to D}\circ \left(\mc{N}^{A\to B}\otimes \mathrm{id}_E\right)\circ \mc{V}^{C\to AE}$,
where superscripts denote input/output systems, $\mathrm{id}$ is the identity map, and $\mc{W},\mc{V}$ are also quantum channels. 
Channel resource measures $\mc{R}:\mc{L}(\mc{H})\to\mathbb{R}$ map a quantum channel to a real number satisfying\\
(R1) Nonnegativity: $\mc{R}(\mc{N})\geq 0$ and $\mc{R}(\mc{M})=0,\forall\mc{M}\in\mc{F}$;\\ 
(R2) Monotonicity: $\mc{R}(\Lambda(\mc{N}))\leq \mc{R}(\mc{N}),\forall\mc{N}\in\mc{L}(\mc{H})$ and $\forall\Lambda\in\mc{O}$.\\
We can construct channel resource monotones with distance quantifiers of channels as
\begin{equation}\label{Eq:channelmonotone1}
	{\mc{R}}(\mc{N})=\min_{\mc{M}\in\mc{F}}{D}(\mc{N},\mc{M}).
\end{equation}
Here, a distance quantifier of two channels ${D}(\mc{N},\mc{M})$ is generally defined by maximizing a distance quantifier of states over all input quantum states~\cite{Leditzky2018approaches},
${D}(\mc{N},\mc{M})=\max_{\rho_{AE}\in\mc{D}(\mc{H}_{AE})}D(\mc{N}_A\otimes \mathrm{id}_E(\rho_{AE}),\mc{M}_A\otimes \mathrm{id}_E(\rho_{AE}))$, where $A$ denotes the system of interest and $E$ is any ancillary system. Examples using this construction include the diamond norm~\cite{Aharonov1998quantum,Watrous2009semi} and entropies of quantum channels~\cite{Cooney2016,yuan2019hypothesis,gour2018entropy}. 


Our proposed framework is a natural mathematical extension of state resource theories, which can be used to characterize general properties of quantum channels. A channel resource theory can be independent of any state resource, such as the quality of a quantum memory~\cite{Rosset2018resource,2018arXiv180903403S,yuan2019robustness}. However, in common scenarios quantum channels are used to manipulate quantum states, and thus it is natural to further extend our framework where the channel resource theory interacts with a state resource theory. As some channels may generate more resource from input states than others, we mainly focus on the resource generation ability of channels. One can also consider general manipulation abilities of channels, such as resource detection~\cite{theurer2018quantifying}, which we leave as future works.

\emph{Interplay with state resource theories.---}Consider a state resource theory $\mathbf{S}=(\Omega, \Phi,\mu)$ with a tensor product structure, i.e., $\phi\otimes \mathrm{id}\in \Phi,\forall \phi\in\Phi$. 
To characterize the state resource generating power, we construct a corresponding channel resource theory $\mathbf{C}=(\mc{F},\mc{O},\mc{R})$. 
As RNG channels cannot increase resource, we define free channels $\mc{F}$ to be RNG channels of $\mathbf{S}$.
We define free super-operations $\mc{O}$ as those quantum super-channels that only use free operations in $\Phi$, i.e., 
  $\Lambda(\mc{N})=\phi_1\circ\left(\mc{N}\otimes \mathrm{id}\right)\circ\phi_2$
for $\phi_1,\phi_2\in\Phi$. The motivation for this configuration is that since free operations can be applied to manipulate states, they should also be allowed in channel manipulation. 
Resource monotones can be defined as Eq.~\eqref{Eq:channelmonotone1} with $\mc{F}=\mathrm{RNG}$.
Alternatively, we introduce a
generalized distance quantifier $D_{\Omega}(\mc{N},\mc{M})=\max_{\omega_{AE}\in\Omega_{AE}}D(\mc{N}_A\otimes \mathrm{id}_E(\omega_{AE}),\mc{M}_A\otimes \mathrm{id}_E(\omega_{AE}))$ with a maximization only over free states. 
Although $D_{\Omega}(\mc{N},\mc{M})$ may be increased under a general super-channel, it defines legitimate channel resource monotones
\begin{equation}
    \mc{R}_{\Omega}(\mc{N})=\min_{\mc{M}\in\mathrm{RNG}}D_{\Omega}(\mc{N},\mc{M}).
\end{equation}
As $D_{\Omega}(\mc{N},\mc{M})\le D(\mc{N},\mc{M})$, we have ${\mc{R}}_{\Omega}(\mc{N})\le {\mc{R}}(\mc{N})$ for the same distance quantifier $D$.
One can further consider an optimization over a general set of states that satisfies certain properties, an interesting open direction of defining general channel resource quantifiers.

Both ${\mc{R}}_{\Omega}(\mc{N})$ and ${\mc{R}}(\mc{N})$ can be regarded as minimizing distances to free channels. 
Alternatively, we can also define channel resource monotones by considering the manipulation ability of channels with respect to the resource of quantum states.
Based on the intuition that a more powerful channel can generate or boost more resource from input states, we define two channel monotones as
\begin{equation}\label{Eq:gbdefinition}
	\begin{aligned}
		\mc{R}_g(\mc{N})=&\max_{\omega_{AE}\in\Omega_{AE}}\mu(\mc{N}_A\otimes \mathrm{id}_E(\omega_{AE})),\\
		\mc{R}_b(\mc{N})=&\max_{\rho_{AE}\in\mathcal{D}(\mathcal{H}_{AE})}\left(\mu(\mathcal{N}_A\otimes \mathrm{id}_E(\rho_{AE}))-\mu(\rho_{AE})\right),
	\end{aligned}
\end{equation}
where $\mu(\rho)=\min_{\omega\in\Omega}D(\rho,\omega)$.
The \emph{resource generating power} $\mc{R}_g(\mc{N})$ characterizes the maximal output resource that can be generated from free input states, while the \emph{resource boosting power} $\mc{R}_b(\mc{N})$ characterizes the maximal boosted resource between the output and input states. Similar concepts have been studied in several specific resource theories, including coherence~\cite{Mani2015cohering,Zanardi2017measures,Zanardi2017coherence,BU20171670,Dana2018resource,Diaz2018usingreusing}, thermodynamics~\cite{Navascues2015nonthermal}, non-Gaussianity~\cite{zhuang2018resource}, and others~\cite{chitambar2018quantum}. For a general resource theory, the resource generating/boosting power has been initially proposed by~\textcite{li2018quantifying} with an optimization only over states of system $A$. 
Different from Ref.~\cite{li2018quantifying}, our definitions consider an optimization with ancillae and more general distance measures. {We summarize the four types of channel resource monotones constructed above as follows and leave the proof in Supplementary Materials.




\begin{theorem}
	For any distance quantifier $D(\rho,\sigma)$ satisfying nonnegativity and data-processing inequality, ${\mc{R}}(\mc{N}),\mc{R}_{\Omega}(\mc{N}),\mc{R}_g(\mc{N}),\mc{R}_b(\mc{N})$ are  channel resource monotones satisfying (R1-2).
\end{theorem}

Interestingly, although ${\mc{R}}(\mc{N})$ and  ${\mc{R}}_{\Omega}(\mc{N})$ are defined with channel distance quantifiers, they are closely related to the state resource manipulation power $\mc{R}_g(\mc{N})$ and $\mc{R}_b(\mc{N})$.

\begin{theorem}
For a state resource theory $\mathbf{S}$ and its corresponding channel resource theory $\mathbf{C}$ as constructed above, we have $\mc{R}_g(\mc{N})\leq \mathcal{R}_{\Omega}(\mathcal{N})$. Furthermore, assuming that $D$ satisfies the triangle inequality, we also have $\mc{R}_b(\mc{N})\le {\mc{R}}(\mc{N})$.
\end{theorem}


\noindent The resource generating power $\mc{R}_g$ is upper bounded by the generalized distance-based resource monotone $\mathcal{R}_{\Omega}$. 
For specific resource theories, such as coherence, the equal sign can also be achieved as shown later in this paper.
Meanwhile, the resource boosting power $\mc{R}_b(\mc{N})$ is also upper bounded by the distance-based monotone ${\mc{R}}(\mc{N})$ under certain assumptions. 
Our results show fundamental connections between the resource theories of quantum channels and states, where the resource manipulation power is upper bounded by distance-based resource monotones, even without specifying the particular resource being studied. Next, we consider operational tasks of channel distillation and dilution  and show their characterizations with the proposed resource monotones.




\emph{Channel distillation and dilution.---}
The key of an operational resource theory is to study resource interconversion, where standard operational tasks include resource distillation and dilution.
For a channel resource theory $\mathbf{C}=(\mc{F},\mc{O},\mc{R})$, we assume that there exists a set of \emph{optimal unit resource channels} $\mathfrak{G}$ which are defined from physical considerations. Then a natural question is how can we convert non-optimal channels to optimal ones and vice versa. For a channel resource theory constructed from a state resource theory, a natural definition of the optimal unit resource channels are those that have maximal resource generating power, i.e. those quantum channels that can generate a maximal unit resource state from certain free states. For example, the Hadamard gate has the maximal resource generating power in the resource theory of coherence. There also exists channel resource theories that do not rely on state resources, including quantum memory~\cite{Rosset2018resource,2018arXiv180903403S,yuan2019robustness} and channel purity, in which unitary channels are regarded as optimal resources.
With the definition of optimal unit resource channels, we define asymptotic quantum channel distillation/dilution based on two ways of using channels.

\noindent\textbf{Definition 1.} 
  \emph{The parallel quantum channel distillation/dilution rate is defined as}
  \begin{equation}\label{defdistill}
\begin{aligned}
  \mc{R}_{\text{distill}}(\mc{N})=&\lim_{\varepsilon\to0^+}\lim_{n\to\infty}\max\big\{R:\exists \mc{G}\in\mathfrak{G},\Lambda\in\mc{O},\\
  &\|\Lambda\left(\mathcal{N}^{\otimes n}\right)-\mc{G}^{\otimes nR}\|_{\diamond}\leq \varepsilon\big\}\\
  \mc{R}_{\text{dilute}}(\mc{N})=&\lim_{\varepsilon\to0^+}\lim_{n\to\infty}\min\big\{R:\exists \mc{G}\in\mathfrak{G},\Lambda\in\mc{O},\\
  &\|\Lambda\left(\mc{G}^{\otimes nR}\right)-\mathcal{N}^{\otimes n}\|_{\diamond}\leq \varepsilon\big\},
\end{aligned}
\end{equation}

\noindent\textbf{Definition 2.} 
  \emph{The iterative quantum channel distillation/dilution rate is defined as}
  \begin{equation}\label{defgeneraldistill}
\begin{aligned}
  \tilde{\mc{R}}_{\text{distill}}(\mc{N})=&\lim_{\varepsilon\to0^+}\lim_{n\to\infty}\max\big\{R:\exists \mc{G}\in\mathfrak{G},\Lambda_i\in\mc{O},\\
  &\|\Lambda_1(\mc{N})\circ\cdots\circ\Lambda_n(\mc{N})-\mc{G}^{\otimes nR}\|_{\diamond}\leq \varepsilon\big\}\\
  \tilde{\mc{R}}_{\text{dilute}}(\mc{N})=&\lim_{\varepsilon\to0^+}\lim_{n\to\infty}\min\big\{R:\exists \mc{G}\in\mathfrak{G},\Lambda_i\in\mc{O},\\
  &\|\Lambda_1(\mc{G})\circ\cdots\circ\Lambda_{nR}(\mc{G})-\mathcal{N}^{\otimes n}\|_{\diamond}\leq \varepsilon\big\}.
\end{aligned}
\end{equation}

\noindent Here $\|\N_A\|_{\diamond}=\max_{\rho_{AE}\in\mc D(\mc H_{AE})}\Tr|\N_A\otimes \mathrm{id}_E(\rho_{AE})|$ denotes the diamond norm of channels.  Note that the iterative protocols are stronger than the parallel ones as the channels are used independently so that outputs of channels can further be inputs of other channels. Therefore we have $\mc{R}_{\text{distill}}(\mc{N})\leq \tilde{\mc{R}}_{\text{distill}}(\mc{N})\le\tilde{\mc{R}}_{\text{dilute}}(\mc{N})\le \mc{R}_{\text{dilute}}(\mc{N})$. However, the distilled or diluted channels after the iterative protocols can only be used in parallel, as their input spaces are tensor producted. If we want to obtain independent channels that can be used iteratively, we can distill the channels one by one. Such a scenario corresponds to one-shot channel distillation and dilution, which will be considered in future works.

Next, we show how resource monotones are related to the asymptotic distillation/dilution rate. For this purpose, we introduce \emph{asymptotic resource measures}, a universal set of channel resource measures which are useful for characterizing asymptotic operational tasks. They are defined as resource monotones that satisfy further requirements besides (R1-2), including: \\
(R3) Normalization: $\mc{R}(\mc{G})=1$ for all $\mc{G}\in\mathfrak{G}$; \\
(R4) Additivity: $\mc{R}(\mc{N}\otimes\mc{M})=\mc{R}(\mc{N})+\mc{R}(\mc{M})$;\\ 
(R$4'$) Subadditivity: $\mc{R}(\mc{N}\otimes\mc{M})\leq\mc{R}(\mc{N})+\mc{R}(\mc{M})$;\\
(R5) Continuity: $|\mc{R}(\mc{N})-\mc{R}(\mc{M})|\leq O(f(\varepsilon)\log d)$ with $\lim_{\varepsilon\to 0^+}f(\varepsilon)=0$, when $\|\mc{N}-\mc{M}\|_{\diamond}\leq\varepsilon$.\\
Here $d$ denotes the dimension of the Hilbert space of the outputs. We note that (R$4'$) is a weaker version of (R4).
The asymptotic resource measures play an important role in parallel channel distillation/dilution.


\begin{theorem}
For any asymptotic resource measure $\mc{R}_{\mathrm{asymp}}(\mc{N})$ satisfying (R1-5), we have 
  \begin{equation}
  	\begin{aligned}
  		\mc{R}_{\mathrm{distill}}(\mc{N})&\leq \mc{R}_{\mathrm{asymp}}(\mc{N})\leq \mc{R}_{\mathrm{dilute}}(\mc{N}).
  	\end{aligned}
  \end{equation}
\end{theorem}




It is in general hard to explicitly determine the distillation/dilution rate as well as constructing asymptotic resource measures without specifying the resource structure.
However, following the spirit of studying the resource generating power of quantum channels, we can focus on the special case where the optimal unit resource channel is uniquely defined to be the channel $\mc{G}_\star(\rho)=\rho_m$, which constantly outputs the maximal unit resource state $\rho_m$ {and is the strongest channel for generating state resources.
This definition aligns with the intuition that channels which generate more powerful resource states are more resourceful. In this case, we are able to characterise the distillation and dilution rates with the resource generating/boosting power of channels in the general setting, which can be applied to all channel resource theories with a corresponding state resource. 
}
Assuming that the state resource monotone $\mu$ satisfies similar requirements as (R1-5), we show that the distillation rates can be explicitly determined. 




\begin{theorem}
When $\mathfrak{G}=\{\mc G_\star\}$, the parallel and iterative channel distillation rates satisfy
\begin{equation}
\begin{aligned}
    \mc{R}_g(\mc{N})\leq&\mc{R}_{\mathrm{distill}}(\mc{N})\leq \lim_{n\to\infty}\frac{1}{n}\mc{R}_g\left(\mc{N}^{\otimes n}\right),\\
    &\tilde{\mc{R}}_{\mathrm{distill}}(\mc{N})\le\mc{R}_b(\mc{N}),
\end{aligned}
\end{equation}
where we assume that $\mc{R}_g$ satisfies (R3), (R5) and $\mc{R}_b$ satisfies (R3). The equal sign $\mc{R}_g(\mc{N})=\mc{R}_{\mathrm{distill}}(\mc{N})$ is achieved if $\mc{R}_g$ also satisfies (R4'), and $\tilde{\mc{R}}_{\mathrm{distill}}(\mc{N})=\mc{R}_b(\mc{N})$ is achieved when the state resource theory is asymptotically reversible. 
\end{theorem}


\noindent Our results verify the intuition that the resource generating/boosting power quantifies how much resource a channel can generate/boost in the asymptotic operational setting. Furthermore, they indicate a possible separation between iterative and prallel protocols in the distillation task, as $\mc R_b$ is generally larger than $\mc R_g$. Meanwhile, we find that the parallel and iterative dilution protocols are actually equivalent, and lower bounded by $\mc{R}_b(\mc{N})$. 






\begin{theorem}
When $\mathfrak{G}=\{\mc G_\star\}$, the parallel and iterative channel dilution rates satisfy
\begin{equation}
    \mc{R}_{\text{dilute}}(\mc{N})= \tilde{\mc{R}}_{\text{dilute}}(\mc{N})\geq \mc{R}_b(\mc{N}),
\end{equation}
where we assume that $\mc{R}_b(\mc{N})$ satisfies (R3) and (R5).
\end{theorem}
\noindent Note that the channel dilution protocol is closely related to channel simulation using state resource, which has been studied for various resource theories~\cite{chitambar2018quantum}. In general, one can also study other manipulation power of channels, where different (maybe also nonequivalent) optimal resource channels can be considered for different purposes. We leave the general case for future work as this work aims to establish the basic resource framework of channels and study its interplay with state resource theories.


Following the quantification of the asymptotic rates, a particular interesting property to study for an operational resource theory is the asymptotic reversibility, which exists in both the resource theories of entanglement and coherence~\cite{brandao2008entanglement,Brandao2010reversible,brandao2015reversible,zhao2018oneshot}. In the context of channel resource theory, asymptotic reversibility is defined as $\mc{R}_{\text{distill}}(\mc{N})=\mc{R}_{\text{dilute}}(\mc{N})$
for all quantum channel $\mc{N}$. 
{It is unclear whether a reversible channel resource theory exists without specifying the resource structure. When focusing on resource generation, our results indicate that the channel resource theory is irreversible as long as there exists a channel $\mc N$ such that $\mc R_b(\mc N)>\lim_{n\to\infty}\frac{1}{n}\mc{R}_g\left(\mc{N}^{\otimes n}\right)$.} While our results contribute as a first step, establishing the reversibility criteria for general channel resource theories requires extensive study on the asymptotic behavior of channel resource quantifiers~\cite{brandao2015reversible,gour2019quantify}. In the following, we give an explicit example in the context of quantum coherence. We show that even though the state resource theory of coherence is reversible, the corresponding channel resource theory is not. 

\emph{Coherence.---}We first review the resource theory of coherence. For a fixed basis $\{\ket{i}\}_{i=0}^{d-1}$ in a $d$-dimensional Hilbert space, free states or incoherent states are those diagonal in the basis, i.e.,  $\mc{I}=\left\{\delta|\delta=\sum_{i=0}^{d-1}\delta_i\ket{i}\bra{i}\right\}$. Meanwhile, maximal resource states are those with a uniform superposition of basis states. For free operations, we consider the maximal set of operations, maximal incoherent operations (MIO)~\cite{aberg2006quantifying}, which map incoherent states to incoherent states. Interestingly, the resource theory of coherence is asymptotically reversible under MIO~\cite{zhao2018oneshot}, characterized by the relative entropy of coherence $C_r(\rho)=\min_{\delta\in\mc{I}}S(\rho\|\delta)$. Here $S(\rho\|\sigma)=\Tr(\rho\log\rho)-\Tr(\rho\log\sigma)$ is the quantum relative entropy.

Now we construct a channel resource theory of coherence to characterize the coherence generating power. This has been partially done in several works \cite{Mani2015cohering,Zanardi2017measures,Zanardi2017coherence,BU20171670,Dana2018resource,Diaz2018usingreusing,theurer2018quantifying,Dana2018resource,Diaz2018usingreusing}, whereas they did not treat channel coherence as an operational resource. Following our resource framework, we define free channels as  resource non-generating channels, i.e., MIO. Free super-operations are a subset of super-channels that transform a quantum channel into another using MIO, $\mc{N}\to\mc{M}_1\circ\left(\mc{N}\otimes \mathrm{id}\right)\circ\mc{M}_2$ for all $\mc{M}_1,\mc{M}_2\in\mathrm{MIO}$. Finally, following our general results, we have four resource monotones
\begin{equation}\nonumber
    \begin{aligned}
{\mc{C}}_{r}(\mc{N})&=\min_{\mc{M}\in\mathrm{MIO}}\max_{\rho\in\mc{D}(\mc{H}_{AE})}S\left(\mc{N}\otimes \mathrm{id}(\rho)\|\mc{M}\otimes \mathrm{id}(\rho)\right),\\ 
\mc{C}_{r,\mc I}(\mc{N})&=\min_{\mc{M}\in\mathrm{MIO}}\max_{\delta\in\mc{I}_{AE}}S\left(\mc{N}\otimes \mathrm{id}(\delta)\|\mc{M}\otimes \mathrm{id}(\delta)\right),\\ 
\mc{C}_{r,g}(\mc{N})&= \max_{\delta\in\mathcal{I}_{AE}}C_r(\mc{N}\otimes \mathrm{id}(\delta)),\\ \mc{C}_{r,b}(\mc{N})&=\max_{\rho\in\mathcal{D}(\mathcal{H}_{AE})}\left(C_r(\mathcal{N}\otimes \mathrm{id}(\rho))-C_r(\rho)\right).
    \end{aligned}
\end{equation}

To further study these four monotones, we utilize the resource structure of coherence, including that it has a resource destroying map~\cite{Liu2017resource}.  
Generalizing Theorem 2, we show that the channel resource monotone $\mc{C}_{r,\mc I}(\mc{N})$ actually equals to the resource generating power $\mc{C}_{r,g}(\mc{N})$,
\begin{equation}
	\mc{C}_{r,\mc I}(\mc{N})=\mc{C}_{r,g}(\mc{N}).
\end{equation}
Furthermore, using the convexity of $C_r$ we have $\mc{C}_{r,\mc I}(\mc{N})=\max_{i}C_r(\mathcal{N}(\ket{i}\bra{i}))$, showing that $\mc{C}_{r,\mc I}(\mc{N})$ is efficiently computable. This simplified expression allows us to prove many useful properties and conclude that $\mc{C}_{r,\mc I}(\mc{N})$ is actually an asymptotic resource measure satisfying (R1-5).
By definition, we have $\mc{C}_{r,\mc I}(\mc{N})\le \mc{C}_{r}(\mc{N})$ and $\mc{C}_{r,g}(\mc{N})\le \mc{C}_{r,b}(\mc{N})$. However, as the relative entropy violates the triangle inequality, we cannot decide the relation between $\mc{C}_{r}(\mc{N})$ and $\mc{C}_{r,b}(\mc{N})$. 

Considering channel distillation/dilution, we define optimal unit resource channels as $\mathfrak{G}=\left\{\mc{G}_\star|\mc{G}_\star(\rho)=\Psi_2\right\}$ with the maximal coherent qubit state $\ket{\Psi_2}=(\ket{0}+\ket{1})/\sqrt{2}$.
As shown in Ref.~\cite{Diaz2018usingreusing}, such a channel can simulate an arbitrary channel that outputs a qubit. 
Denote the asymptotic parallel (iterative) channel coherence distillation/dilution rate as $\mc{C}_{\mathrm{distill}}$ ($\tilde{\mc{C}}_{\mathrm{distill}}$) and $\mc{C}_{\mathrm{dilute}}$ ($\tilde{\mc{C}}_{\mathrm{dilute}}$), respectively. As an immediate consequence of Theorem 3 and 4, it follows that $\mc{C}_{\mathrm{distill}}(\mc{N})\leq \mc{C}_{r,\mc I}(\mc{N}) \leq \mc{C}_{\mathrm{dilute}}(\mc{N})$ and $\tilde{\mc{C}}_{\mathrm{distill}}(\mc{N})=\mc{C}_{r,b}(\mc{N})$.
For parallel distillation, we are also able to show that the equal sign is achieved, due to the additivity of $\mc{C}_{r,\mc I}(\mc{N})$.


\begin{corollary}The parallel and iterative asymptotic distillation rates of channel coherence are $\mc{C}_{\mathrm{distill}}(\mc{N})=\mc{C}_{r,\mc I}(\mc{N})$ and $\tilde{\mc{C}}_{\mathrm{distill}}(\mc{N})=\mc{C}_{r,b}(\mc{N})$, respectively.
\end{corollary}


\noindent Note that there exists quantum channel $\mc{N}$ such that $\mc{C}_{r,b}(\mc{N})> \mc{C}_{r,\mc I}(\mc{N})$ (see Supplementary Materials for the explicit example) and therefore $\mc{C}_{\mathrm{dilute}}(\mc{N})=\tilde{\mc{C}}_{\mathrm{dilute}}(\mc{N})\ge \tilde{\mc{C}}_{\mathrm{distill}}(\mc{N})>\mc{C}_{\mathrm{distill}}(\mc{N})$, indicating the irreversibility of the operational resource theory of channel coherence.


\begin{corollary}
The channel resource theory of coherence  is asymptotically irreversible.
\end{corollary}


For the dilution part, besides the lower bound given in Theorem 5, we consider the (smooth) max entropy of channel coherence, 
\begin{equation}\label{channelmaxentropy}
\begin{aligned}
  \mc{C}_{\mathrm{max}}^\varepsilon(\mc{N})=&\min_{\|\mc{N}-\mc{N}'\|_{\diamond}\leq\varepsilon}\log\min\big\{\lambda:\exists\mc{M}\in\mathrm{MIO},\\
  &\mc{N}'\leq\lambda\mc{M}\big\},
\end{aligned}
\end{equation}
which is used to characterize one-shot channel simulation using coherence~\cite{Diaz2018usingreusing}. We show that the asymptotic dilution rate is equal to its regularized version.


\begin{theorem}
The asymptotic dilution rate of channel coherence equals to the regularized max entropy of channel coherence,
\begin{equation}
    \mc{C}_{\mathrm{dilute}}(\mc{N})=\tilde{\mc{C}}_{\mathrm{dilute}}(\mc{N})=\mc{C}_{\mathrm{max}}^\infty(\mc{N}),
\end{equation}
where
$ \mc{C}_{\mathrm{max}}^\infty(\mc{N})=\lim_{\varepsilon\to 0^+}\lim_{n\to\infty}\frac{1}{n}\mc{C}_{\mathrm{max}}^\varepsilon(\mc{N}^{\otimes n})$.
\end{theorem}



\noindent Calculating the limit requires the development of a channel-analogy of the generalized Quantum Stein's Lemma~\cite{Brandao2010generalization}, which is beyond the scope of this paper and is left as an important future work.

\emph{Discussion.---}Our work introduces a general framework for the resource theory of quantum channels and study its interplay with the resource theory of states.  We also introduce the operational tasks of channel distillation and dilution and study their characterizations with channel resource measures.  We consider coherence as an example and find the channel resource theory of coherence asymptotically irreversible. 
Future works can extend our results to other quantum resources that have specific resource structure. 
As a higher level resource, channel resource theory can focus on the manipulation power of state resources, including coherence~\cite{Dana2018resource,Diaz2018usingreusing}, entanglement~\cite{Kaur_2017}, discord, magic~\cite{seddon2019quantifying, wang2019quantifying}, and thermodynamics. More interestingly, channel resources can be independent of state resources, such as quantum memory~\cite{Rosset2018resource,2018arXiv180903403S,yuan2019robustness} and channel purity. 
Meanwhile, previous works have mainly focused on certain properties of channels without considering quantum channels as operational resources. 
Completing the operational resource framework of channels requires extensive future works, such as studying basic entropic quantifiers of channels~\cite{Cooney2016,Leditzky2018approaches,yuan2019hypothesis,gour2018entropy} and investigating channel conversions~\cite{gour2018comparison}.
We hope our work establishes the basic framework and can inspire future works of channel resource theories. \\


\begin{acknowledgments}
\emph{Acknowledgement.---}We acknowledge Qi Zhao for insightful discussions. 
Yunchao Liu was supported by the National Natural Science Foundation of China Grants No.~11875173 and No.~11674193, and the National Key R\&D Program of China Grants No.~2017YFA0303900 and No.~2017YFA0304004. 
Xiao Yuan was supported by the EPSRC National Quantum Technology Hub in Networked Quantum Information Technology (EP/M013243/1).\\ 

\emph{Note added.---}Recently we became aware of a closely related work by Zi-Wen Liu and Andreas Winter~\cite{liu2019resource} who independently propose a similar framework for channel resource theories and further study the robustness measure and its role in resource erasure.
\end{acknowledgments}

\bibliographystyle{apsrev4-1}
\bibliography{BibChannel}

\newpage
\onecolumngrid
\appendix

\section{Properties of Channel Resource Monotones}
We prove that the four proposed channel resource quantifiers satisfy the properties of resource monotones. Here Results~\ref{monotone1}-\ref{monotone3} correspond to Theorem 1 in the main text.

\begin{result}\label{monotone1}
For any distance quantifier $D(\rho,\sigma)$ satisfying nonnegativity and data-processing inequality, ${\mc{R}}(\mc{N})$ is a channel resource monotone.
\end{result}

\begin{proof}
  Recall the definition
  \begin{equation}\label{defRN}
    {\mc{R}}(\mc{N})=\min_{\mc{M}\in\mc{F}}{D}(\mc{N},\mc{M})
  \end{equation}
where ${D}(\mc{N},\mc{M})=\max_{\rho_{AE}\in\mc{D}(\mc{H}_{AE})}D(\mc{N}_A\otimes \mathrm{id}_E(\rho_{AE}),\mc{M}_A\otimes \mathrm{id}_E(\rho_{AE}))$. From the nonnegativity of $D(\rho,\sigma)$ and Eq.~\eqref{defRN} we know that $\mc{R}(\mc{N})\geq 0$ while $\mc{R}(\mc{M})=0$ for all $\mc{M}\in\mc{F}$. For the monotonicity requirement (R2), we have
\begin{equation}\label{RNmonotonicity}
  \begin{aligned}
  \mc{R}(\mc{N})&=\min_{\mc{M}\in\mc{F}}\max_{\rho_{AE}\in\mc{D}(\mc{H}_{AE})}D(\mc{N}_A\otimes \mathrm{id}_E(\rho_{AE}),\mc{M}_A\otimes \mathrm{id}_E(\rho_{AE}))\\
  &\geq \min_{\mc{M}\in\mc{F}}\max_{\rho_{AE'E}\in\mc{D}(\mc{H}_{AE'E})}D(\mc{N}_A\otimes \mathrm{id}_{E'}\otimes \mathrm{id}_{E}(\rho_{AE'E}),\mc{M}_A\otimes \mathrm{id}_{E'}\otimes \mathrm{id}_{E}(\rho_{AE'E}))\\
  &\geq \min_{\mc{M}\in\mc{F}}\max_{\rho_{AE'E}\in\mc{D}(\mc{H}_{AE'E})}D((\mc{N}_A\otimes \mathrm{id}_{E'})\circ\mc{V}\otimes \mathrm{id}_{E}(\rho_{AE'E}),(\mc{M}_A\otimes \mathrm{id}_{E'})\circ\mc{V}\otimes \mathrm{id}_{E}(\rho_{AE'E}))\\
  &\geq \min_{\mc{M}\in\mc{F}}\max_{\rho_{AE'E}\in\mc{D}(\mc{H}_{AE'E})}D(\mc{U}\circ(\mc{N}_A\otimes \mathrm{id}_{E'})\circ\mc{V}\otimes \mathrm{id}_{E}(\rho_{AE'E}),\mc{U}\circ(\mc{M}_A\otimes \mathrm{id}_{E'})\circ\mc{V}\otimes \mathrm{id}_{E}(\rho_{AE'E}))\\
  &\geq \min_{\mc{M}\in\mc{F}}\max_{\rho_{AE'E}\in\mc{D}(\mc{H}_{AE'E})}D(\mc{U}\circ(\mc{N}_A\otimes \mathrm{id}_{E'})\circ\mc{V}\otimes \mathrm{id}_{E}(\rho_{AE'E}),\mc{M}_{AE'}\otimes \mathrm{id}_{E}(\rho_{AE'E}))\\
  &=\mc{R}\left(\mc{U}\circ(\mc{N}\otimes \mathrm{id})\circ\mc{V}\right).
  \end{aligned}
\end{equation}
Here, the second line appends another system $E'$ which by definition is less optimal than the original system $AE$. The third line follows from the fact that choosing from the image of $\mc{V}\otimes \mathrm{id}_E$ is less optimal than choosing from all states. The fourth line follows from the data-processing inequality of $D(\rho,\sigma)$. The fifth line follows from the fact that $\mc{U}\circ(\mc{M}_A\otimes \mathrm{id}_{E'})\circ\mc{V}\in \mc F$ and we are optimizing over a larger set of $\mc{F}$.
\end{proof}

\begin{result}\label{monotone2}
  For any distance quantifier $D(\rho,\sigma)$ satisfying nonnegativity and data-processing inequality, ${\mc{R}}_{\Omega}(\mc{N})$ is a channel resource monotone.
\end{result}

\begin{proof}
  The nonnegativity (R1) follows similarly as above. For the monotonicity requirement, note that the definition of free super-operations are $\Lambda(\mc{N})=\phi_1\circ\left(\mc{N}\otimes \mathrm{id}\right)\circ\phi_2$ for $\phi_1,\phi_2\in\Phi$. Then,
  \begin{equation}\label{RomegaNmonotonicity}
  \begin{aligned}
  \mc{R}_\Omega(\mc{N})&=\min_{\mc{M}\in\mc{F}}\max_{\omega_{AE}\in\Omega_{AE}}D(\mc{N}_A\otimes \mathrm{id}_E(\omega_{AE}),\mc{M}_A\otimes \mathrm{id}_E(\omega_{AE}))\\
  &\geq \min_{\mc{M}\in\mc{F}}\max_{\omega_{AE'E}\in\Omega_{AE'E}}D(\mc{N}_A\otimes \mathrm{id}_{E'}\otimes \mathrm{id}_{E}(\omega_{AE'E}),\mc{M}_A\otimes \mathrm{id}_{E'}\otimes \mathrm{id}_{E}(\omega_{AE'E}))\\
  &\geq \min_{\mc{M}\in\mc{F}}\max_{\omega_{AE'E}\in\Omega_{AE'E}}D((\mc{N}_A\otimes \mathrm{id}_{E'})\circ\phi_1\otimes \mathrm{id}_{E}(\omega_{AE'E}),(\mc{M}_A\otimes \mathrm{id}_{E'})\circ\phi_1\otimes \mathrm{id}_{E}(\omega_{AE'E}))\\
  &\geq \min_{\mc{M}\in\mc{F}}\max_{\omega_{AE'E}\in\Omega_{AE'E}}D(\phi_2\circ(\mc{N}_A\otimes \mathrm{id}_{E'})\circ\phi_1\otimes \mathrm{id}_{E}(\omega_{AE'E}),\phi_2\circ(\mc{M}_A\otimes \mathrm{id}_{E'})\circ\phi_1\otimes \mathrm{id}_{E}(\omega_{AE'E}))\\
  &\geq \min_{\mc{M}\in\mc{F}}\max_{\omega_{AE'E}\in\Omega_{AE'E}}D(\phi_2\circ(\mc{N}_A\otimes \mathrm{id}_{E'})\circ\phi_1\otimes \mathrm{id}_{E}(\omega_{AE'E}),\mc{M}_{AE'}\otimes \mathrm{id}_{E}(\omega_{AE'E}))\\
  &=\mc{R}_\Omega\left(\phi_2\circ(\mc{N}\otimes \mathrm{id})\circ\phi_1\right).
  \end{aligned}
\end{equation}
Here, similar as above, the second line appends another system $E'$ which by definition is less optimal than the original system $AE$. The third line follows from the fact that choosing from the image of $\phi_1\otimes \mathrm{id}_E$ is less optimal than choosing from all free states. The fourth line follows from the data-processing inequality of $D(\rho,\sigma)$. The fifth line follows from the fact that $\phi_2\circ(\mc{M}_A\otimes \mathrm{id}_{E'})\circ\phi_1\in \mc F$ and we are optimizing over a larger set of $\mc{F}$.
\end{proof}

\begin{result}\label{monotone3}
  For a state resource monotone $\mu$ satisfying (S1-2), the resource generating powers $$\mc{R}_g(\mc{N})=\max_{\omega_{AE}\in\Omega_{AE}}\mu(\mc{N}_A\otimes \mathrm{id}_E(\omega_{AE}))$$ and $$\mc{R}_b(\mc{N})=\max_{\rho_{AE}\in\mathcal{D}(\mathcal{H}_{AE})}\left(\mu(\mathcal{N}_A\otimes \mathrm{id}_E(\rho_{AE}))-\mu(\rho_{AE})\right)$$ are channel resource monotones satisfying (R1-2).
\end{result}

\begin{proof}
  We start with $\mc{R}_g(\mc{N})$. Nonnegativity (R1) follows from the nonnegativity of $\mu$. For $\mc{M}\in\mathrm{RNG}$, we have $\mc{M}\otimes \mathrm{id}(\omega)\in\Omega$ and thus $\mc{R}_g(\mc{M})=0$. For monotonicity (R2), we have
  \begin{equation}\label{Rgmonotonicity}
    \begin{aligned}
    \mc{R}_g(\mc{N})&=\max_{\omega_{AE}\in\Omega_{AE}}\mu(\mc{N}_A\otimes \mathrm{id}_E(\omega_{AE}))\\
    &\geq \max_{\omega_{AE'E}\in\Omega_{AE'E}}\mu(\mc{N}_A\otimes \mathrm{id}_{E'}\otimes \mathrm{id}_E(\omega_{AE'E}))\\
    &\geq \max_{\omega_{AE'E}\in\Omega_{AE'E}}\mu((\mc{N}_A\otimes \mathrm{id}_{E'})\circ \phi_1\otimes \mathrm{id}_E(\omega_{AE'E}))\\
    &\geq \max_{\omega_{AE'E}\in\Omega_{AE'E}}\mu(\phi_2\circ (\mc{N}_A\otimes \mathrm{id}_{E'})\circ \phi_1\otimes \mathrm{id}_E(\omega_{AE'E}))\\
    &=\mc{R}_g(\phi_2\circ(\mc{N}\otimes \mathrm{id})\circ\phi_1).
    \end{aligned}
  \end{equation}
  The above equation follows a similar logic as before. The second line follows by appending another ancillary system which is less optimal. The third line follows from the fact that $\phi\otimes \mathrm{id}(\omega)\in\Omega$. The fourth line follows from the monotonicity of $\mu$.

  Now we consider $\mc{R}_b(\mc{N})$. For any quantum channel $\mc{N}$, we have
  \begin{equation}
    \begin{aligned}
    \mc{R}_b(\mc{N})&=\max_{\rho_{AE}\in\mathcal{D}(\mathcal{H}_{AE})}\left(\mu(\mathcal{N}_A\otimes \mathrm{id}_E(\rho_{AE}))-\mu(\rho_{AE})\right)\\
    &\geq \max_{\omega_{AE}\in\Omega_{AE}}\left(\mu(\mathcal{N}_A\otimes \mathrm{id}_E(\omega_{AE}))-\mu(\omega_{AE})\right)\\
    &=\mc{R}_g(\mc{N})\\
    &\geq 0.
    \end{aligned}
  \end{equation}
  On the other hand, for any $\mc{M}\in\mathrm{RNG}$,
    \begin{equation}
    \begin{aligned}
    \mc{R}_b(\mc{M})&=\max_{\rho_{AE}\in\mathcal{D}(\mathcal{H}_{AE})}\left(\mu(\mathcal{M}_A\otimes \mathrm{id}_E(\rho_{AE}))-\mu(\rho_{AE})\right)\\
    &\leq \max_{\rho_{AE}\in\mathcal{D}(\mathcal{H}_{AE})}\left(\mu(\rho_{AE})-\mu(\rho_{AE})\right)\\
    &=0,
    \end{aligned}
  \end{equation}
therefore $\mc{R}_b(\mc{M})=0$. For monotonicity (R2), we have
\begin{equation}
  \begin{aligned}
  \mc{R}_b(\mc{N})&=\max_{\rho_{AE}\in\mathcal{D}(\mathcal{H}_{AE})}\left(\mu(\mathcal{N}_A\otimes \mathrm{id}_E(\rho_{AE}))-\mu(\rho_{AE})\right)\\
  &\geq \max_{\rho_{AE'E}\in\mathcal{D}(\mathcal{H}_{AE'E})}\left(\mu(\mathcal{N}_A\otimes \mathrm{id}_{E'}\otimes \mathrm{id}_E(\rho_{AE'E}))-\mu(\rho_{AE'E})\right)\\
  &\geq \max_{\rho_{AE'E}\in\mathcal{D}(\mathcal{H}_{AE'E})}\left(\mu((\mathcal{N}_A\otimes \mathrm{id}_{E'})\circ\phi_1\otimes \mathrm{id}_E(\rho_{AE'E}))-\mu(\phi_1\otimes \mathrm{id}_E(\rho_{AE'E}))\right)\\
  &\geq \max_{\rho_{AE'E}\in\mathcal{D}(\mathcal{H}_{AE'E})}\left(\mu((\mathcal{N}_A\otimes \mathrm{id}_{E'})\circ\phi_1\otimes \mathrm{id}_E(\rho_{AE'E}))-\mu(\rho_{AE'E})\right)\\
  &\geq \max_{\rho_{AE'E}\in\mathcal{D}(\mathcal{H}_{AE'E})}\left(\mu(\phi_2\circ(\mathcal{N}_A\otimes \mathrm{id}_{E'})\circ\phi_1\otimes \mathrm{id}_E(\rho_{AE'E}))-\mu(\rho_{AE'E})\right)\\
  &=\mc{R}_b(\phi_2\circ(\mc{N}\otimes \mathrm{id})\circ\phi_1).
  \end{aligned}
\end{equation}
Here, the second line follows by appending another ancillary system which is less optimal. The third line follows from the fact that the image space of $\phi_1\otimes \mathrm{id}_E$ is smaller than all states. The fourth and fifth line both uses monotonicity of $\mu$.
\end{proof}

\section{Interplay with State Resource Theories}
We present the interplay between channel and state resource theories by showing the general relationship between distance-based resource monotones and resource generating powers. Besides Results~\ref{Rgupperbound} and~\ref{Rbupperbound} which correspond to Theorem~2 in the main text, we also present a lower bound of a different form of $\mc R_g(\mc N)$ in Remark~\ref{Rglowerbound}.

\begin{result}\label{Rgupperbound}
For a state resource theory $\mathbf{S}=(\Omega, \Phi,\mu)$ and its corresponding channel resource theory $\mathbf{C}=(\mathrm{RNG},\mc{O},\mc{R})$ as constructed in the main text, we have
 \begin{equation}
 	\mc{R}_g(\mc{N})\leq \mathcal{R}_{\Omega}(\mathcal{N}).
 \end{equation}

\end{result}

\begin{proof}
  With the definition of $\mc{R}_\Omega(\mc{N})$, we have
\begin{equation}
    \begin{aligned}
    \mc{R}_\Omega(\mc{N})&=\min_{\mc{M}\in\mathrm{RNG}}\max_{\omega_{AE}\in\Omega_{AE}}D(\mc{N}_A\otimes \mathrm{id}_E(\omega_{AE}),\mc{M}_A\otimes \mathrm{id}_E(\omega_{AE}))\\
    &\geq \min_{\mc{M}\in\mathrm{RNG}}\max_{\omega_{AE}\in\Omega_{AE}}\min_{\sigma\in\Omega_{AE}}D(\mc{N}_A\otimes \mathrm{id}_E(\omega_{AE}),\sigma)\\
    &=\max_{\omega_{AE}\in\Omega_{AE}}\min_{\sigma\in\Omega_{AE}}D(\mc{N}_A\otimes \mathrm{id}_E(\omega_{AE}),\sigma)\\
    &=\max_{\omega_{AE}\in\Omega_{AE}}\mu\left(\mc{N}_A\otimes \mathrm{id}_E(\omega_{AE})\right)\\
    &=\mc{R}_g(\mc{N}).
    \end{aligned}
\end{equation}
Here, the second line follows by adding an additional minimization, and the third line is because $\mc{M}$ does not appear in the objective. The fourth line follows from the definition of $\mu$.

\end{proof}

\begin{result}\label{Rbupperbound}
For a state resource theory $\mathbf{S}=(\Omega, \Phi,\mu)$ and its corresponding channel resource theory $\mathbf{C}=(\mathrm{RNG},\mc{O},\mc{R})$, when the distance quantifier of states $D$ in the definition of $\mu$ satisfies the triangle inequality, we have $$\mc{R}_b(\mc{N})\leq {\mc{R}}(\mc{N}).$$
\end{result}

\begin{proof}
  For the distance quantifier $D$, we assume that it satisfies triangle inequality, which is $D(\rho,\sigma)\leq D(\rho,\delta)+D(\delta,\sigma)$ for all states $\rho,\sigma,\delta$. Then, it follows that
  \begin{equation}\label{triangle}
    \begin{aligned}
    \mc{R}(\mc{N})&=\min_{\mc{M}\in\mathrm{RNG}}\max_{\rho_{AE}\in\mc{D}(\mc{H}_{AE})}D(\mc{N}_A\otimes \mathrm{id}_E(\rho_{AE}),\mc{M}_A\otimes \mathrm{id}_E(\rho_{AE}))\\
    &\geq \min_{\mc{M}\in\mathrm{RNG}}\max_{\rho_{AE}\in\mc{D}(\mc{H}_{AE})}\left(D(\mc{N}_A\otimes \mathrm{id}_E(\rho_{AE}),\delta)-D(\mc{M}_A\otimes \mathrm{id}_E(\rho_{AE}),\delta)\right)\\
    &\geq \max_{\rho_{AE}\in\mc{D}(\mc{H}_{AE})}\min_{\mc{M}\in\mathrm{RNG}}\left(D(\mc{N}_A\otimes \mathrm{id}_E(\rho_{AE}),\delta)-D(\mc{M}_A\otimes \mathrm{id}_E(\rho_{AE}),\delta)\right).
    \end{aligned}
  \end{equation}
Here, the second line follows from triangle inequality and the third line follows from minimax theorem. Note that Eq.~\eqref{triangle} holds for any $\delta$, and thus also holds for the choice $$\delta^*=\arg\min_{\omega\in\Omega_{AE}}D(\mc{M}_A\otimes \mathrm{id}_E(\rho_{AE}),\omega).$$
Then,
  \begin{equation}
    \begin{aligned}
    \mc{R}(\mc{N})&\geq \max_{\rho_{AE}\in\mc{D}(\mc{H}_{AE})}\min_{\mc{M}\in\mathrm{RNG}}\left(D(\mc{N}_A\otimes \mathrm{id}_E(\rho_{AE}),\delta^*)-D(\mc{M}_A\otimes \mathrm{id}_E(\rho_{AE}),\delta^*)\right)\\
    &=\max_{\rho_{AE}\in\mc{D}(\mc{H}_{AE})}\min_{\mc{M}\in\mathrm{RNG}}\left(D(\mc{N}_A\otimes \mathrm{id}_E(\rho_{AE}),\delta^*)-\mu(\mc{M}_A\otimes \mathrm{id}_E(\rho_{AE}))\right)\\
    &\geq\max_{\rho_{AE}\in\mc{D}(\mc{H}_{AE})}\min_{\mc{M}\in\mathrm{RNG}}\left(D(\mc{N}_A\otimes \mathrm{id}_E(\rho_{AE}),\delta^*)-\mu(\rho_{AE})\right)\\
    &\geq\max_{\rho_{AE}\in\mc{D}(\mc{H}_{AE})}\left(\min_{\omega\in\Omega_{AE}}D(\mc{N}_A\otimes \mathrm{id}_E(\rho_{AE}),\omega)-\mu(\rho_{AE})\right)\\
    &=\max_{\rho_{AE}\in\mc{D}(\mc{H}_{AE})}\left(\mu(\mc{N}_A\otimes \mathrm{id}_E(\rho_{AE}))-\mu(\rho_{AE})\right)\\
    &=\mc{R}_b(\mc{N}).
    \end{aligned}
  \end{equation}
Here, the second line follows from the definition of $\mu$. The third line follows from monotonicity of $\mu$. The fourth line replaces the minimization of $\mc{M}$ to a larger set (note that $\delta^*\in\Omega_{AE}$). The fifth line also follows from the definition of $\mu$.
\end{proof}

\begin{remark}\label{Rglowerbound}
We can also characterize $\mc{R}_g(\mc{N})$ with the definition of resource destroying (RD) channels, which are quantum channels that map any state to a free state. For a $\lambda\in\mathrm{RD}$, we define $\mu_\lambda(\rho)=D(\rho,\lambda(\rho))$, which may not be a resource monotone as it can violate both requirements. We consider it to be a useful quantifier due to its closed-from expression and $\mu(\rho)=\min_{\lambda\in\mathrm{RD}}\mu_\lambda(\rho)$ (Recall the definition $\mu(\rho)=\min_{\omega\in\Omega}D(\rho,\omega)$).
In particular, we can upper bound  $\mc{R}_\Omega(\mc{N})$ as follows,

 \begin{equation}
 	 \mathcal{R}_{\Omega}(\mathcal{N})\leq\max_{\omega_{AE}\in\Omega_{AE}}\mu_{\lambda_A}(\mc{N}_A\otimes \mathrm{id}_E(\omega_{AE})),
 \end{equation}
 where $\mu_{\lambda_A}(\mc{N}_A\otimes \mathrm{id}_E(\omega_{AE})) = D(\mc{N}_A\otimes \mathrm{id}_E(\omega_{AE}), \lambda_A(\mc{N}_A\otimes \mathrm{id}_E(\omega_{AE})))$. Note that here the resource destroying channel $\lambda$ only acts on system $A$. The proof is as follows:

  Note that the channel $\mathcal{E}^*=\lambda_A\circ\mathcal{N}_A$ for any resource destroying channel $\lambda_A$ is a resource non-generating channel. Then we have
\begin{equation}
    \begin{aligned}
    \mc{R}_\Omega(\mc{N})&=\min_{\mc{M}\in\mathrm{RNG}}\max_{\omega_{AE}\in\Omega_{AE}}D(\mc{N}_A\otimes \mathrm{id}_E(\omega_{AE}),\mc{M}_A\otimes \mathrm{id}_E(\omega_{AE}))\\
    &\leq\max_{\omega_{AE}\in\Omega_{AE}}D(\mc{N}_A\otimes \mathrm{id}_E(\omega_{AE}),\lambda_A\circ\mathcal{N}_A\otimes \mathrm{id}_E(\omega_{AE}))\\
    &=\max_{\omega_{AE}\in\Omega_{AE}}D(\mc{N}_A\otimes \mathrm{id}_E(\omega_{AE}),\lambda_A\circ(\mathcal{N}_A\otimes \mathrm{id}_E)(\omega_{AE}))\\
    &=\max_{\omega_{AE}\in\Omega_{AE}}\mu_{\lambda_A}(\mc{N}_A\otimes \mathrm{id}_E(\omega_{AE})).
    \end{aligned}
\end{equation}

For resource theories with free states admitting a product structure, i.e., $\omega_{AE}=\sum_ip_i\omega_{A,i}\otimes\omega_{E,i},\forall \omega_{AE}\in \Omega_{AE}$, we have 
\begin{equation}
    \mu_{\lambda_A}(\mc{N}_A\otimes \mathrm{id}_E(\omega_{AE})) =\mu_{\lambda}(\mc{N}_A\otimes \mathrm{id}_E(\omega_{AE})).
\end{equation}
Furthermore, when the state resource theory has a resource destroying channel $\lambda$ that satisfies $\mu(\rho) = D(\rho,\lambda(\rho))$, the equal signs are achieved with $\mc{R}_g(\mc{N})= \mathcal{R}_{\Omega}(\mathcal{N})$. 
\end{remark}

\section{Channel Distillation and Dilution}
We first prove the general relationship between parallel and iterative protocols (Lemma~\ref{paralleliterative}) and then study the universal role of asymptotic resource measures in these tasks (Result~\ref{asympbound}, Theorem~3). Following an Assumption on the optimal resource channels, we prove bounds on the distillation (Results~\ref{standarddistillbound} and~\ref{generaldistillbound}, Theorem~4) and dilution rates (Result~\ref{dilutionandlowerbound}, Theorem~5) using resource generating powers.

\begin{figure}[ht]
    \centering
    \includegraphics[width=10cm]{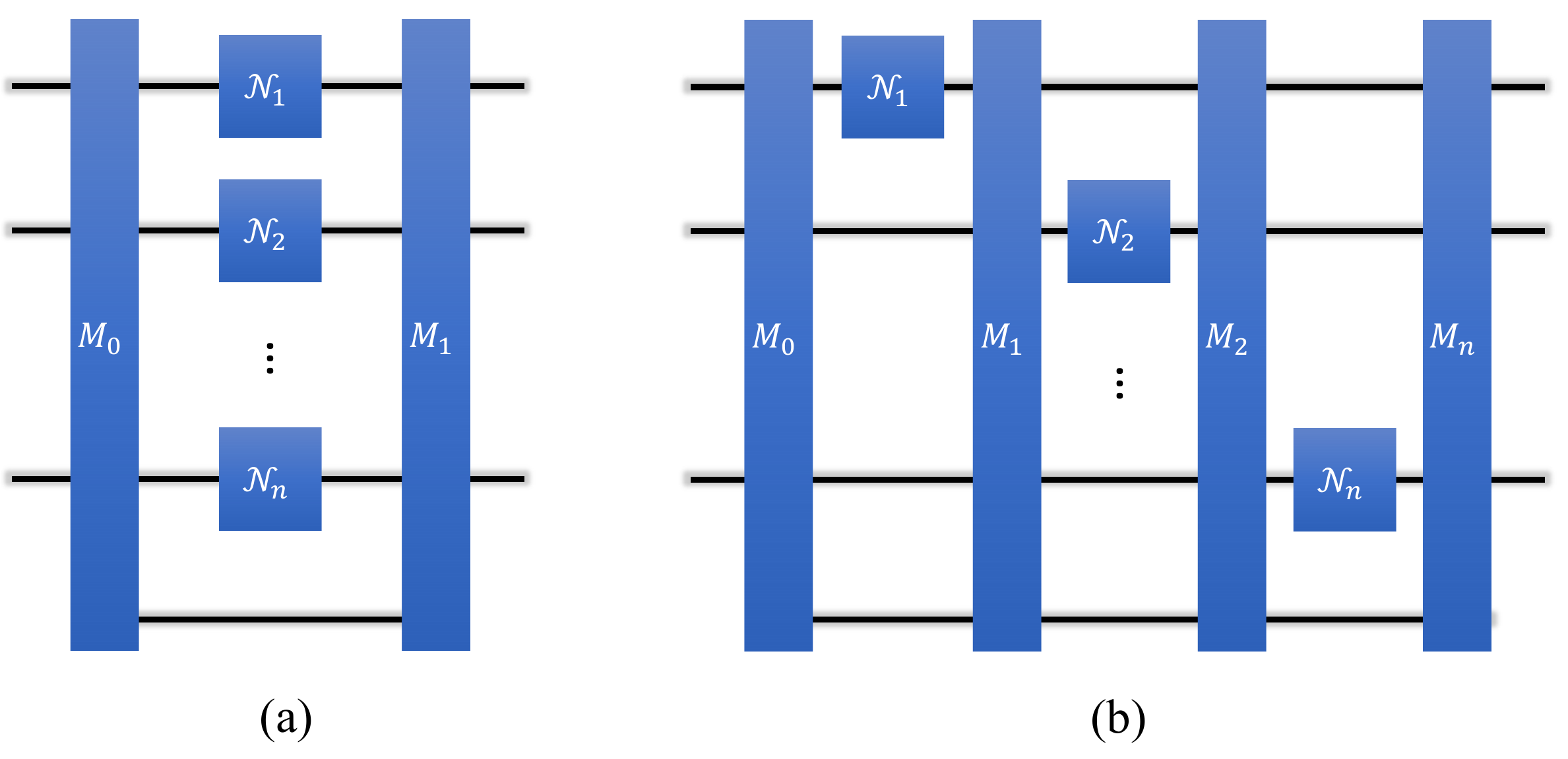}
    \caption{Different schemes for quantum channel distillation and dilution. (a)The parallel protocol, where quantum channels are concatenated in parallel. (b)The iterative protocol, where quantum channels are manipulated sequentially.}\label{figuredistill}
  \end{figure}

\begin{lemma}\label{paralleliterative}
The iterative distillation/dilution protocol contains the parallel distillation/dilution protocol, that is, $\mc{R}_{\text{distill}}(\mc{N})\leq \tilde{\mc{R}}_{\text{distill}}(\mc{N})\le\tilde{\mc{R}}_{\text{dilute}}(\mc{N})\le \mc{R}_{\text{dilute}}(\mc{N})$.
\end{lemma}
\begin{proof}
  We prove this Lemma by showing a simple construction. For an arbitrary parallel protocol in Fig.~\ref{figuredistill}(a), we construct the iterative protocol in Fig.~\ref{figuredistill}(b). We require that the dimensions between $M_0,M_1,\cdots,M_n$ are large enough such that each quantum channel $\mc{N}_i$ can act on different subsystems. Now, we can choose $M_1,\cdots,M_{n-1}$ to be identity, and the result is the parallel protocol in Fig.~\ref{figuredistill}(a). In other words, every parallel protocol is also an iterative protocol, so by definition we have the inequalities.

\end{proof}

\begin{result}\label{asympbound}
For any asymptotic resource measure $\mc{R}_{\mathrm{asymp}}$ for a channel resource theory and for any quantum channel $\mc{N}$, we have
  \begin{equation}
  	\begin{aligned}
  		\mc{R}_{\mathrm{distill}}(\mc{N})&\leq \mc{R}_{\mathrm{asymp}}(\mc{N})\leq \mc{R}_{\mathrm{dilute}}(\mc{N}).
  	\end{aligned}
  \end{equation}
\end{result}
\begin{proof}
  For the left hand side, suppose we have a parallel distillation protocol such that for a large $n$, $$\|\phi_1\circ(\mc{N}^{\otimes n}\otimes \mathrm{id})\circ\phi_2-\mathcal{G}^{\otimes nR}\|_\diamond\leq\varepsilon.$$ Then,
  \begin{equation}
    \begin{aligned}
    \mc{R}_{\mathrm{asymp}}(\mc{N})&=\frac{1}{n}\mc{R}_{\mathrm{asymp}}\left(\mc{N}^{\otimes n}\right)\\
    &\geq \frac{1}{n}\mc{R}_{\mathrm{asymp}}\left(\phi_1\circ(\mc{N}^{\otimes n}\otimes \mathrm{id})\circ\phi_2\right)\\
    &\geq \frac{1}{n}\left(\mc{R}_{\mathrm{asymp}}\left(\mathcal{G}^{\otimes nR}\right)-O(n\varepsilon)\right)\\
    &=R-O(\varepsilon).
    \end{aligned}
  \end{equation}

Here, the first line follows from additivity (R4). The second line follows from monotonicity (R2). The third line follows from continuity (R5). Taking the limit of $\varepsilon\to 0^+$ and $n\to\infty$, we have $\mc{R}_{\mathrm{asymp}}(\mc{N})\geq R$. Since this holds for any distillation protocol, we conclude that $\mc{R}_{\mathrm{distill}}(\mc{N})\leq \mc{R}_{\mathrm{asymp}}(\mc{N})$.

For the right hand side, suppose we have a parallel dilution protocol such that for a large $n$, $$\|\phi_1\circ(\mathcal{G}^{\otimes nR}\otimes \mathrm{id})\circ\phi_2-\mc{N}^{\otimes n}\|_\diamond\leq\varepsilon.$$ Then,
  \begin{equation}
    \begin{aligned}
    \mc{R}_{\mathrm{asymp}}(\mc{N})&=\frac{1}{n}\mc{R}_{\mathrm{asymp}}\left(\mc{N}^{\otimes n}\right)\\
    &\leq \frac{1}{n}\left(\mc{R}_{\mathrm{asymp}}\left(\phi_1\circ(\mathcal{G}^{\otimes nR}\otimes \mathrm{id})\circ\phi_2\right)+O(n\varepsilon)\right)\\
    &\leq \frac{1}{n}\left(\mc{R}_{\mathrm{asymp}}\left(\mathcal{G}^{\otimes nR}\right)+O(n\varepsilon)\right)\\
    &=R+O(\varepsilon).
    \end{aligned}
  \end{equation}

Here, the first line follows from additivity (R4). The second line follows from continuity (R5). The third line follows from monotonicity (R2). Taking the limit of $\varepsilon\to 0^+$ and $n\to\infty$, we have $\mc{R}_{\mathrm{asymp}}(\mc{N})\leq R$. Since this holds for any dilution protocol, we conclude that $\mc{R}_{\mathrm{asymp}}(\mc{N})\leq \mc{R}_{\mathrm{dilute}}(\mc{N})$.
\end{proof}

\noindent\textbf{Assumptions.} In the following, we assume that the optimal resource channels $\mathfrak{G}$ are uniquely defined to be the constant channel, $\mathfrak{G}=\{\mc G_\star\}$, where $$\mc{G}_\star(\rho)=\rho_m$$ is the channel that constantly outputs the maximal unit resource state $\rho_m$.


\noindent\textbf{Properties of state resource measures.} In analogy of the requirements for the asymptotic resource measures of channels, we also give additional requirements for the state resource measure $\mu$ as well.
\begin{itemize}
    \item (S3) Normalization: $\mu(\rho_m)=1$ for all maximal unit resource state $\rho_m$;
    \item (S4) Additivity: $\mu(\rho\otimes\sigma)=\mu(\rho)+\mu(\sigma)$;
    \item (S4') Subadditivity: $\mu(\rho\otimes\sigma)\leq\mu(\rho)+\mu(\sigma)$;
    \item (S5) Continuity: $|\mu(\rho)-\mu(\sigma)|\leq O(f(\varepsilon)\log d)$ with $\lim_{\varepsilon\to 0^+}f(\varepsilon)=0$, when $\|\rho-\sigma\|_1\leq\varepsilon$.
\end{itemize}

Here $d$ denotes the dimension of the Hilbert space of the states. Now, we note the following properties for the channel monotones $\mc{R}_g$ and $\mc{R}_b$.

\begin{lemma}\label{superadditive}
  Recall the definitions \begin{equation}
	\begin{aligned}
		\mc{R}_g(\mc{N})=&\max_{\omega_{AE}\in\Omega_{AE}}\mu(\mc{N}_A\otimes \mathrm{id}_E(\omega_{AE})),\\
		\mc{R}_b(\mc{N})=&\max_{\rho_{AE}\in\mathcal{D}(\mathcal{H}_{AE})}\left(\mu(\mathcal{N}_A\otimes \mathrm{id}_E(\rho_{AE}))-\mu(\rho_{AE})\right),
	\end{aligned}
\end{equation}
when $\mu$ is additive, we conclude that $\mc{R}_g$ and $\mc{R}_b$ is actually superadditive, i.e., $\mc{R}_g(\mc{N}\otimes\mc{M})\geq \mc{R}_g(\mc{N})+\mc{R}_g(\mc{M})$, $\mc{R}_b(\mc{N}\otimes\mc{M})\geq \mc{R}_b(\mc{N})+\mc{R}_b(\mc{M})$ for all quantum channels $\mc{N},\mc{M}$.
\end{lemma}

\begin{proof}
  We start with $\mc{R}_g(\mc{N})$. For arbitrary quantum channels $\mc{N},\mc{M}$,
  \begin{equation}
      \begin{aligned}
        \mc{R}_g(\mc{N}\otimes\mc{M})&=\max_{\omega_{ABE}\in\Omega_{ABE}}\mu(\mc{N}_A\otimes\mc{M}_B\otimes \mathrm{id}_E(\omega_{ABE}))\\
        &\geq \max_{\omega_{ABEE'}\in\Omega_{ABEE'}}\mu(\mc{N}_A\otimes\mc{M}_B\otimes \mathrm{id}_E\otimes \mathrm{id}_{E'}(\omega_{ABEE'}))\\
        &\geq \max_{\omega_{AE}\otimes\delta_{BE'}\in\Omega_{ABEE'}}\mu(\mc{N}_A\otimes\mc{M}_B\otimes \mathrm{id}_E\otimes \mathrm{id}_{E'}(\omega_{AE}\otimes\delta_{BE'}))\\
        &=\max_{\omega_{AE}\otimes\delta_{BE'}\in\Omega_{ABEE'}}\mu\left(\mc{N}_A\otimes \mathrm{id}_E(\omega_{AE})\otimes \mc{M}_B\otimes\mathrm{id}_{E'}(\delta_{BE'})\right)\\
        &=\max_{\omega_{AE}\otimes\delta_{BE'}\in\Omega_{ABEE'}}\left(\mu\left(\mc{N}_A\otimes \mathrm{id}_E(\omega_{AE})\right)+\mu\left(\mc{M}_B\otimes\mathrm{id}_{E'}(\delta_{BE'})\right)\right)\\
        &=\mc{R}_g(\mc{N})+\mc{R}_g(\mc{M}).
      \end{aligned}
  \end{equation}
Here the second line follows by appending an additional ancillary system, and the third line follows by only choosing product states which is less optimal. The fifth line follows from the additivity of $\mu$.
 
Similarly, for $\mc{R}_b(\mc{N})$ we have
  \begin{equation}
      \begin{aligned}
        \mc{R}_b(\mc{N}\otimes\mc{M})&=\max_{\rho_{ABE}\in\mathcal{D}(\mathcal{H}_{ABE})}\left(\mu(\mathcal{N}_A\otimes\mc{M}_B\otimes \mathrm{id}_E(\rho_{ABE}))-\mu(\rho_{ABE})\right)\\
        &\geq \max_{\rho_{ABEE'}\in\mathcal{D}(\mathcal{H}_{ABEE'})}\left(\mu(\mathcal{N}_A\otimes\mc{M}_B\otimes \mathrm{id}_E\otimes\mathrm{id}_{E'}(\rho_{ABEE'}))-\mu(\rho_{ABEE'})\right)\\
        &\geq \max_{\rho_{AE}\otimes\sigma_{BE'}\in\mathcal{D}(\mathcal{H}_{ABEE'})}\left(\mu(\mc{N}_A\otimes \mathrm{id}_E(\rho_{AE})\otimes \mc{M}_B\otimes\mathrm{id}_{E'}(\rho_{BE'}))-\mu(\rho_{AE}\otimes\sigma_{BE'})\right)\\
        &= \max_{\rho_{AE}\otimes\sigma_{BE'}\in\mathcal{D}(\mathcal{H}_{ABEE'})}\left(\mu(\mc{N}_A\otimes \mathrm{id}_E(\rho_{AE}))+\mu( \mc{M}_B\otimes\mathrm{id}_{E'}(\rho_{BE'}))-\mu(\rho_{AE})-\mu(\sigma_{BE'})\right)\\
        &=\mc{R}_b(\mc{N})+\mc{R}_b(\mc{M}).
      \end{aligned}
  \end{equation}
Here the second line follows by appending an additional ancillary system, the third line follows by only choosing product states which is less optimal, and the fourth line follows from the additivity of $\mu$.
\end{proof}

In the following, let
\begin{equation}
    \mc{R}_g(\mc{N})=\max_{\omega_{AE}\in\Omega_{AE}}\mu_d(\mc{N}_A\otimes \mathrm{id}_E(\omega_{AE})),
\end{equation}
where $\mu_d$ is the asymptotic distillation rate of the state resource theory.

\begin{result}\label{standarddistillbound}
The parallel channel distillation rate is
$$\mc{R}_g(\mc{N})\leq\mc{R}_{\mathrm{distill}}(\mc{N})\leq \lim_{n\to\infty}\frac{1}{n}\mc{R}_g\left(\mc{N}^{\otimes n}\right),$$ when $\mc{R}_g$ satisfies (R3) and (R5). The equal sign is achieved if $\mc{R}_g$ also satisfies (R4').
\end{result}

\begin{proof}
  First we prove the lower bound $\mc{R}_{\mathrm{distill}}(\mc{N})\geq\mc{R}_g(\mc{N})$. Let $R=\mc{R}_g(\mc{N})$, then $\exists \omega\in\Omega_{AE}$ such that $$\mu_d(\mc{N}\otimes \mathrm{id}(\omega))=R.$$
  Under the assumptions, we only need to construct a distillation protocol that generates $\mc{G}_\star$. For a large $n$, we construct an asymptotic distillation protocol $\phi_1\circ\left(\mc{N}^{\otimes n}\otimes \mathrm{id}\right)\circ \phi_2\approx\mc{G}_\star^{\otimes nR}$. First of all, let $\phi_2$ be a quantum channel that constantly outputs $\omega^{\otimes n}$. The state before $\phi_1$ is $\sigma^{\otimes n}$ where $\sigma=\mc{N}\otimes \mathrm{id}(\omega)$. By the definition of $\mu_d$, there exists $\phi_1$ such that $\phi_1\left(\sigma^{\otimes n}\right)\approx \rho_m^{\otimes nR}$. The resulting channel $\phi_1\circ\left(\mc{N}^{\otimes n}\otimes \mathrm{id}\right)\circ \phi_2$ is close to $\mc{G}_\star^{\otimes nR}$ in terms of diamond norm. By definition of channel distillation rate, we have $\mc{R}_{\mathrm{distill}}(\mc{N})\geq\mc{R}_g(\mc{N})$.

  For the upper bound, suppose we have a parallel distillation protocol such that for a large $n$, $$\|\phi_1\circ(\mc{N}^{\otimes n}\otimes \mathrm{id})\circ\phi_2-\mathcal{G}^{\otimes nR}\|_\diamond\leq\varepsilon.$$ Then,
  \begin{equation}
    \begin{aligned}
 \frac{1}{n}\mc{R}_g\left(\mc{N}^{\otimes n}\right) &\geq \frac{1}{n}\mc{R}_g\left(\phi_1\circ(\mc{N}^{\otimes n}\otimes \mathrm{id})\circ\phi_2\right)\\
    &\geq \frac{1}{n}\left(\mc{R}_g\left(\mathcal{G}^{\otimes nR}\right)-O(n\varepsilon)\right)\\
    &=R-O(\varepsilon).
    \end{aligned}
  \end{equation}
Here, the first line follows from monotonicity (R2). The second line follows from continuity (R5). Taking the limit of $\varepsilon\to 0^+$ and $n\to\infty$, we have $\lim_{n\to\infty}\frac{1}{n}\mc{R}_g\left(\mc{N}^{\otimes n}\right)\geq R$. Since this holds for any distillation protocol, we conclude that $\lim_{n\to\infty}\frac{1}{n}\mc{R}_g\left(\mc{N}^{\otimes n}\right)\geq\mc{R}_{\mathrm{distill}}(\mc{N})$.

When $\mc{R}_g$ is subadditive, we have $\frac{1}{n}\mc{R}_g\left(\mc{N}^{\otimes n}\right)\leq \mc{R}_g\left(\mc{N}\right)$, thus the equal signs are achieved.
\end{proof}

\begin{result}\label{generaldistillbound}
When $\mc{R}_b$ satisfies (R3) and $\mu$ satisfies (S3) and (S5), the iterative channel distillation rate is
$\tilde{\mc{R}}_{\mathrm{distill}}(\mc{N})\le\mc{R}_b(\mc{N})$.
The equal sign is achieved when the state resource theory is asymptotically reversible and $\mu=\mu_d$ is the asymptotic distillation/dilution rate.
\end{result}
\begin{proof}
  We first prove the general bound $\tilde{\mc{R}}_{\mathrm{distill}}(\mc{N})\le\mc{R}_b(\mc{N})$. For a sufficiently large $n$, suppose there is a distillation protocol such that $$\|\phi_{n+1}\circ\left(\mc{N}\otimes \mathrm{id}_n\right)\circ \phi_n\circ\left(\mc{N}\otimes \mathrm{id}_{n-1}\right)\circ\cdots\circ\phi_1-\mc{G}^{\otimes nR}\|_\diamond\leq\varepsilon,$$then $\exists\omega\in\Omega$ such that $$\mu\left(\phi_{n+1}\circ\left(\mc{N}\otimes \mathrm{id}_n\right)\circ \phi_n\circ\left(\mc{N}\otimes \mathrm{id}_{n-1}\right)\circ\cdots\circ\phi_1(\omega)\right)\geq n(R-O(\varepsilon)).$$ Let $\rho_1=\phi_1(\omega)$ and $\rho_n=\phi_n\circ(\mc{N}\otimes \mathrm{id}_{n-1})(\rho_{n-1})$. Then
  \begin{equation*}
    \begin{aligned}
        n(R-O(\varepsilon)) &\leq \mu(\rho_{n+1})\\
        &=\sum_{k=1}^{n}\mu(\rho_{k+1})-\mu(\rho_{k})\\
        &=\sum_{k=1}^{n}\mu\left(\phi_{k+1}\circ(\mc{N}\otimes \mathrm{id}_{k})(\rho_{k})\right)-\mu(\rho_{k})\\
        &\leq \sum_{k=1}^{n}\mu\left(\mc{N}\otimes \mathrm{id}_{k}(\rho_{k})\right)-\mu(\rho_{k})\\
        &\leq \sum_{k=1}^{n}\max_{\rho\in\mathcal{D}(\mathcal{H}_{AE})}\left(\mu(\mathcal{N}\otimes \mathrm{id}(\rho))-\mu(\rho)\right)\\
        &=n\mc{R}_b(\mc{N}).
    \end{aligned}
  \end{equation*}
Here, the third line is by definition of $\rho_{k+1}$, the fourth line is by monotonicity of $\mu$, and the fifth line follows by choosing the optimal $\rho_k$. Taking the asymptotic limit, we have $\tilde{\mc{R}}_{\mathrm{distill}}(\mc{N})\leq \mc{R}_b(\mc{N})$.

When the state resource theory is asymptotically reversible, we can also prove a lower bound $\tilde{\mc{R}}_{\mathrm{distill}}(\mc{N})\geq \mc{R}_b(\mc{N})$, by constructing a distillation protocol that achieves $\mc{R}_b(\mc{N})$. The distillation protocol is described as follows.
   \begin{enumerate}
     \item For any input state, $\phi_1$ outputs a fixed free state $\omega^{\otimes t}$. Let $\sigma=\mc{N}(\omega)$. Through $t$ uses of $\mc{N}$, we obtain the state $\sigma^{\otimes t}$.
     \item We distill the maximal resource state from $\sigma^{\otimes t}$ with rate $\mu_d(\sigma)$ and obtain the state $\rho_m^{\otimes t\mu_d(\sigma)}$.
     \item We dilute the maximal resource state $\rho_m^{\otimes t\mu_d(\sigma)}$ to a state $\rho^{\otimes w}$ where $w=t\frac{\mu_d(\sigma)}{\mu_d(\rho)}$.
     \item Through $w$ usage of $\mc{N}$, we obtain the state $\tilde{\rho}^{\otimes w}$ where $\tilde{\rho}=\mc{N}\otimes \mathrm{id}(\rho)$.
     \item We distill $\tilde{\rho}^{\otimes w}$ into $\rho_m^{\otimes w\mu_d(\tilde{\rho})}$.
     \item We dilute $\rho_m^{\otimes w\mu_d(\rho)}$ into $\rho^{\otimes w}$, saving resource $\rho_m^{\otimes w\left(\mu_d(\tilde{\rho})-\mu_d(\rho)\right)}$ for output.
     \item Repeat Step 4-6 for $k$ times.
   \end{enumerate}
   Overall, the resource distillation rate is $\frac{kw\left(\mu_d(\tilde{\rho})-\mu_d(\rho)\right)}{kw+t}\to\left(\mu_d(\tilde{\rho})-\mu_d(\rho)\right)$ as $k\to\infty$. Since the protocol holds for any $\rho$, we can take the supremum of $\rho$ and obtain distillation rate $\mc{R}_b(\mc{N})$.
\end{proof}

\begin{result}\label{dilutionandlowerbound}
For any channel resource theory satisfying the Assumption, we have $\mc{R}_{\text{dilute}}(\mc{N})= \tilde{\mc{R}}_{\text{dilute}}(\mc{N})$. Furthermore, if $\mc{R}_b(\mc{N})$ satisfies (R3) and (R5) and $\mu$ satisfies (S3) and (S4), we also have $\mc{R}_{\text{dilute}}(\mc{N})= \tilde{\mc{R}}_{\text{dilute}}(\mc{N})\geq \mc{R}_b(\mc{N})$.
\end{result}
\begin{proof}
  Under the assumptions, we can convert any dilution protocol to a protocol that only uses the resource channel $\mc{G}_\star$ which outputs the maximal resource state from any input. As all parallel dilution protocols are contained in iterative dilution protocols, we only need to show that any iterative dilution protocol can reduce to a parallel one. Now, observe the following property of $\mc{G}_\star$,
  \begin{equation}\label{Gstarproperty}
    \left(\mc{G}_{\star A}\otimes \mathrm{id}_B\right)\circ\phi_{AB}(\rho_{AB})=\rho_m\otimes\phi_B'(\rho_B)=(\mathrm{id}_A\otimes\phi_B')\circ(G_{\star A}\otimes \mathrm{id}_B)(\rho_{AB}),
  \end{equation}
which means that we can commute a quantum channel from before $\mc{G}_\star$ to afterwards. Consider an iterative dilution protocol in Fig.~\ref{figuredistill}(b), where we can commute $\mc{G}_2\circ\mc{M}_1$ to $\mc{M}_1'\circ\mc{G}_2$. Then we combine $\mc{M}_1'$ with $\mc{M}_2$ and continue the operation. In the end, we result in a protocol with no intermediate channels between the resource channels. We can furthermore assume that these channels are used in parallel, since sequential concatenation is effectively useless. Thus we have reduced a iterative dilution protocol to a parallel protocol. In conclusion, we have $\mc{R}_{\text{dilute}}(\mc{N})= \tilde{\mc{R}}_{\text{dilute}}(\mc{N})$.

Now we prove the lower bound. Since the two types of dilution protocol are equivalent, we only consider the parallel protocol. Suppose there is a parallel dilution protocol such that for large $n$,
\begin{equation}
  \|\phi_2\circ\left(\mc{G}_\star^{\otimes nR}\otimes \mathrm{id}\right)\circ\phi_1-\mc{N}^{\otimes n}\|_\diamond\leq\varepsilon.
\end{equation}
Then 
\begin{equation}
\begin{aligned}
  \mc R_b(\N^{\otimes n}) &\le \mc R_b\left(\phi_2\circ\left(\mc{G}_\star^{\otimes nR}\otimes \mathrm{id}\right)\circ\phi_1\right)+O(\varepsilon n),\\
  &= \max_{\rho}\mu\left(\left(\phi_2\circ\left(\mc{G}_\star^{\otimes nR}\otimes \mathrm{id}\right)\circ\phi_1\right)\otimes\mathrm{id}(\rho)\right)-\mu(\rho)+O(\varepsilon n),\\
  &= \max_{\rho}\mu\left(\left(\phi_2\circ\left(\mc{G}_\star^{\otimes nR}\otimes \mathrm{id}\right)\right)(\rho')\right)-\mu(\rho)+O(\varepsilon n),\\
  &\le \max_{\rho}\mu\left(\left(\mc{G}_\star^{\otimes nR}\otimes \mathrm{id}\right)(\rho')\right)-\mu(\rho)+O(\varepsilon n),\\
  &\le \max_{\rho}nR+\mu(\rho')-\mu(\rho)+O(\varepsilon n),\\
  &\le \max_{\rho}nR+O(\varepsilon n).\\
\end{aligned}
\end{equation}
Here the first line is due to the continuity of $\mc R_b(\mc{N})$; in the third line, we denote $\rho'=\phi_1\otimes \mathrm{id}(\rho)$; the last three lines are due to the monotonicity and additivity of $\mu$.

Recall from Lemma~\ref{superadditive} that the additivity of $\mu$ implies the superadditivity of $\mc{R}_b(\mc{N})$, which gives us
\begin{equation}
  \mc{R}_b(\mc{N})\leq \frac{1}{n}\mc{R}_b\left(\mc{N}^{\otimes n}\right)\leq R+O(\varepsilon).
\end{equation}
Taking the asymptotic limit $\varepsilon\to 0^+$ and $n\to\infty$, we obtain the lower bound.

\end{proof}

\section{Applications in the Resource Theory of Coherence}
We present the applications of our results in the resource theory of coherence. We first study the properties of the measures of channel coherence (Lemma~\ref{propertycoherence}). In the operational resource theory of channel coherence, we show that the optimal resource channels can be relaxed to a larger set (Lemma~\ref{equivalent}). By applying our general results and using the resource structure of coherence, we show that the distillation rates are exactly equal to the resource generating powers (Result~\ref{coherencedistillation}, Corollary~1), and then conclude that the channel resource theory of coherence is asymptotically irreversible (Result~\ref{coherenceirreversible}, Corollary~2). Finally, we give an exact characterization of the dilution rate with the regularization of the max entropy of channel coherence (Result~\ref{coherencedilution}, Theorem~6).

\begin{lemma}\label{propertycoherence}
  $\mc{C}_{r,\mc I}(\mc{N})=\mc{C}_{r,g}(\mc{N})$ is an asymptotic resource measure satisfying (R1-5).
\end{lemma}
\begin{proof}
  First of all, $\mc{C}_{r,\mc I}(\mc{N})=\mc{C}_{r,g}(\mc{N})$ follows from Result~\ref{Rgupperbound} and Remark~\ref{Rglowerbound}.

  Normalization (R3) follows from the normalization property of $C_r(\rho)$. For additivity (R4), We denote the overall space of incoherent states as $\mathcal{I}$, where the space of incoherent states for the input system of $\mathcal{N},\mathcal{M}$ is $\mathcal{I}_1,\mathcal{I}_2$, respectively. It follows from the definition that $\mathcal{I}_1\otimes \mathcal{I}_2\subseteq\mathcal{I}$. First we have
\begin{equation}
    \begin{aligned}
        \mc{C}_{r,\mc{I}}(\mathcal{N\otimes M})&=\max_{\delta\in\mathcal{I}}C_r(\mathcal{N}\otimes\mathcal{M}(\delta))\\
        &\geq \max_{\delta_1\otimes\delta_2\in\mathcal{I}_1\otimes\mathcal{I}_2}C_r(\mathcal{N}\otimes\mathcal{M}(\delta_1\otimes\delta_2))\\
        &= \max_{\delta_1\otimes\delta_2\in\mathcal{I}_1\otimes\mathcal{I}_2}C_r(\mathcal{N}(\delta_1)\otimes\mathcal{M}(\delta_2))\\
        &=\mc{C}_{r,\mc{I}}(\mathcal{N})+\mc{C}_{r,\mc{I}}(\mathcal{M}).
    \end{aligned}
\end{equation}
Here, the second line is from the fact that $\mathcal{I}_1\otimes \mathcal{I}_2\subseteq\mathcal{I}$, and the fourth line is from the additivity of $C_r(\rho)$.

Next, notice that any incoherent state $\delta\in\mathcal{I}$ can be expressed as $\delta=\sum_i p_i\delta_{1,i}\otimes\delta_{2,i}$, where $\delta_{1,i}\in\mathcal{I}_1$ and $\delta_{2,i}\in\mathcal{I}_2$. Then
\begin{equation}
    \begin{aligned}
        \mc{C}_{r,\mc{I}}(\mathcal{N\otimes M})&=\max_{\delta\in\mathcal{I}}C_r(\mathcal{N}\otimes\mathcal{M}(\delta))\\
        &=\max_{\delta=\sum_i p_i\delta_{1,i}\otimes\delta_{2,i}}C_r(\mathcal{N}\otimes\mathcal{M}(\delta))\\
        &=\max_{\delta=\sum_i p_i\delta_{1,i}\otimes\delta_{2,i}}C_r\left(\sum_ip_i \mathcal{N}(\delta_{1,i})\otimes\mathcal{M}(\delta_{2,i})\right)\\
        &\leq \max_{\delta=\sum_i p_i\delta_{1,i}\otimes\delta_{2,i}}\sum_i p_i C_r\left(\mathcal{N}(\delta_{1,i})\otimes\mathcal{M}(\delta_{2,i})\right)\\
        &\leq \max_{\delta_1\otimes\delta_2}C_r\left(\mathcal{N}(\delta_{1})\otimes\mathcal{M}(\delta_{2})\right)\\
        &=\mc{C}_{r,\mc{I}}(\mathcal{N})+\mc{C}_{r,\mc{I}}(\mathcal{M}).
    \end{aligned}
\end{equation}
The fourth line is from the convexity of $C_r(\rho)$, and the last line uses the additivity of $C_r(\rho)$. From the above two equations, we conclude that $\mc{C}_{r,\mc{I}}(\mathcal{N\otimes M})=\mc{C}_{r,\mc{I}}(\mathcal{N})+\mc{C}_{r,\mc{I}}(\mathcal{M})$.

For continuity (R5), since $\mc{C}_{r,\mc{I}}(\mathcal{N})=\max_{i}C_r(\mathcal{N}(\ket{i}\bra{i}))$, we can assume that $\mc{C}_{r,\mc{I}}(\mathcal{N})=C_r(\mathcal{N}(\ket{i}\bra{i}))$ and $\mc{C}_{r,\mc{I}}(\mathcal{M})=C_r(\mathcal{M}(\ket{j}\bra{j}))$. We further assume that $C_r(\mathcal{N}(\ket{i}\bra{i}))\geq C_r(\mathcal{M}(\ket{j}\bra{j}))$. Then
\begin{equation}
    \begin{aligned}
    |\mc{C}_{r,\mc{I}}(\mathcal{N})-\mc{C}_{r,\mc{I}}(\mathcal{M})|&=|C_r(\mathcal{N}(\ket{i}\bra{i}))-C_r(\mathcal{M}(\ket{j}\bra{j}))|\\
    &\leq |C_r(\mathcal{N}(\ket{i}\bra{i}))-C_r(\mathcal{M}(\ket{i}\bra{i}))|+|C_r(\mathcal{M}(\ket{j}\bra{j}))-C_r(\mathcal{M}(\ket{i}\bra{i}))|\\
    &\leq 2|C_r(\mathcal{N}(\ket{i}\bra{i}))-C_r(\mathcal{M}(\ket{i}\bra{i}))|.
    \end{aligned}
\end{equation}
By definition of the diamond norm, we know that $\|\mc{N}-\mc{M}\|_\diamond\leq\varepsilon$ implies $\|\mathcal{N}(\ket{i}\bra{i})-\mathcal{M}(\ket{i}\bra{i})\|_1\leq\varepsilon$. Using the continuity of the relative entropy of coherence, the result is obtained, which is $$|\mc{C}_{r,\mc{I}}(\mathcal{N})-\mc{C}_{r,\mc{I}}(\mathcal{M})|\leq O(\varepsilon\log d).$$
\end{proof}

For the channel resource theory of coherence, we define the optimal resource channels as those who can generate the maximal resource state from some free state, i.e., 
$$\mathfrak{G}=\left\{\mc{G}|\exists\delta\in\mc{I},\mc{G}(\delta)=\Psi_2\right\}.$$
The above definition is a generalization of the Assumption made in the previous Section. Below we show that they are actually equivalent in the resource theory of coherence, therefore our general results still hold under this new definition of the optimal resource channels.

\begin{lemma}\label{equivalent}
  The set of optimal unit resource channels $\mathfrak{G}=\left\{\mc{G}|\exists\delta\in\mc{I},\mc{G}(\delta)=\Psi_2\right\}$ is equivalent with the Assumption.
\end{lemma}
\begin{proof}
To prove that they are equivalent, we only need to show the following Conditions:
\begin{enumerate}
  \item All channels in $\mathfrak{G}$ are equivalent: $\forall\mc{G}_1,\mc{G}_2\in\mathfrak{G}$, $\exists\phi_1,\phi_2\in\Phi$ such that $\mc{G}_2=\phi_1\circ\left(\mc{G}_1\otimes \mathrm{id}\right)\circ\phi_2$.
  \item The channel $\mc{G}_\star(\rho)=\rho_m$ that constantly outputs the maximal unit resource state $\rho_m$ belongs to $\mathfrak{G}$.
\end{enumerate}

By definition, Condition 2 is satisfied, i.e., the constant channel $\mc{G}_\star(\rho)=\Psi_2$ belongs to $\mathfrak{G}$. For Condition 1, we only need to show that every optimal unit resource channel is equivalent to the constant channel $\mc{G}_\star$, that is, $\forall\mc{G}\in\mathfrak{G}$, there exists $\mc{M}_1,\mc{M}_2,\mc{M}_3,\mc{M}_4\in\mathrm{MIO}$ such that
\begin{enumerate}
    \item (a) $\mc{M}_1\circ(\mc{G}\otimes\mathrm{id}_1)\circ\mc{M}_2=\mc{G}_\star$;
    \item (b) $\mc{M}_3\circ(\mc{G}_\star\otimes\mathrm{id}_2)\circ\mc{M}_4=\mc{G}$.
\end{enumerate}

For the first part (a), suppose $\mc{G}(\delta)=\Psi_2$, we set $\mathrm{id}_1$ to be one dimension, and set $\mc{M}_2=\mathrm{id}$ and $\mc{M}_1(\cdot)=\delta$, which simulates $\mc{G}_\star$. For the second part (b), let $\mc{M}_3(\rho)=\sigma\otimes\rho$ for an arbitrary incoherent state $\sigma$, and at the input of $\mc{M}_4$ we have the state $\Psi_2\otimes \rho$. For $\mc{M}_4$, we invoke the channel simulation protocol~\cite{Dana2018resource,Diaz2018usingreusing}, which states that any quantum channel can be implemented by a MIO using a resource state $\Psi_d$ where $d$ equals to the output dimension of the channel, and we let $\mc{M}_4$ to be such a MIO. The overall protocol simulates $\mc{G}$ using $\mc{G}_\star$. We also note that this protocol holds for all quantum channels, not limited to the optimal channels $\mc{G}\in\mathfrak{G}$.
\end{proof}

\begin{result}\label{coherencedistillation}
The parallel and iterative asymptotic distillation rates of channel coherence are $\mc{C}_{\mathrm{distill}}(\mc{N})=\mc{C}_{r,\mc I}(\mc{N})$ and $\tilde{\mc{C}}_{\mathrm{distill}}(\mc{N})=\mc{C}_{r,b}(\mc{N})$, respectively.
\end{result}
\begin{proof}
  These results are direct corollaries of Result~\ref{standarddistillbound} and Result~\ref{generaldistillbound}, as $\mc{C}_{r,\mc I}(\mc{N})$ is an asymptotic channel resource measure and the resource theory of coherence is reversible under MIO.
\end{proof}

\begin{result}\label{coherenceirreversible}
The channel resource theory of coherence is irreversible.
\end{result}
\begin{proof}
From Result~\ref{dilutionandlowerbound} we know that $\mc{C}_{\mathrm{dilute}}(\mc{N})\geq\mc{C}_{r,b}(\mc{N})$. Since $\mc{C}_{r,b}(\mc{N})$ is in general greater than $\mc{C}_{r,\mc I}(\mc{N})$ as verified by our numerical simulation, the conclusion is obtained. For example, considering the unitary channel $\mc{U}(\rho)=U\rho U^\dag$ with $U = e^{-i\sigma_y*\pi/10}$, we have $\mc{C}_{r,\mc I}(\mc{U})=0.4545$ and $\mc{C}_{r,b}(\mc{U})\ge0.5684$. Here
$\sigma_y$ is the Pauli-$Y$ matrix. We plot the calculation in Fig.~\ref{Fig:simulation}.
\end{proof}

\begin{figure}[ht]\centering
  \includegraphics[width=0.5\linewidth]{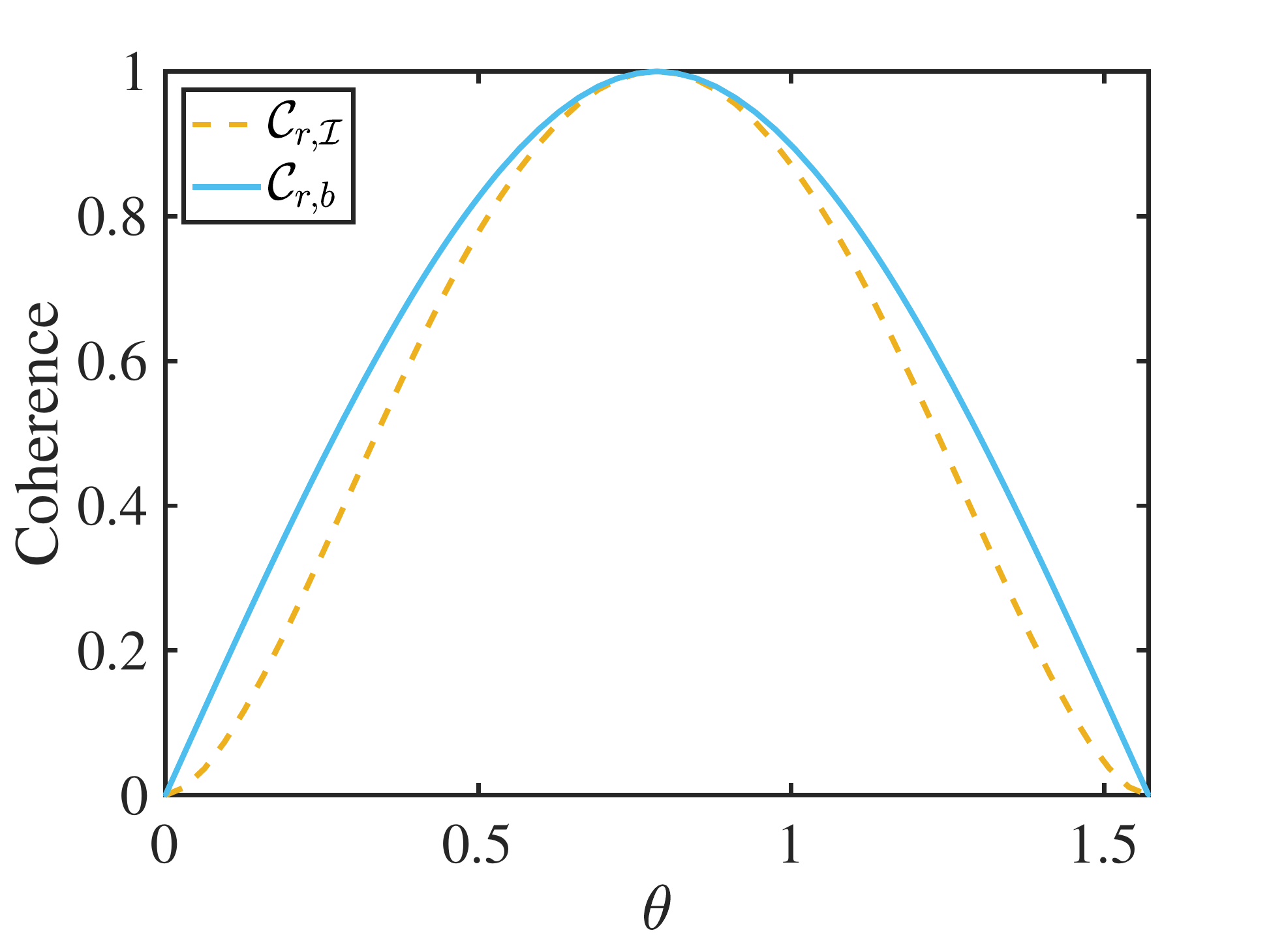}
  \caption{Numerical calculation for the channel coherence of $\mc{U}(\rho)=U\rho U^\dag$ with $U = e^{-i\theta\sigma_y}$. The result for $\mc{C}_{r,\mc I}(\mc{U})$ is accurate due to its simplified expression obtained in Lemma~\ref{propertycoherence}, while the result for $\mc{C}_{r,b}(\mc{U})$ is only a lower bound. }\label{Fig:simulation}
\end{figure}

\begin{result}\label{coherencedilution}
The asymptotic dilution rate of channel coherence equals to the regularized max entropy of channel coherence, $\mc{C}_{\mathrm{dilute}}(\mc{N})=\tilde{\mc{C}}_{\mathrm{dilute}}(\mc{N})=\mc{C}_{\mathrm{max}}^\infty(\mc{N})$, where
  \begin{equation}\label{channelmaxentropyregularized}
  \mc{C}_{\mathrm{max}}^\infty(\mc{N})=\lim_{\varepsilon\to 0^+}\lim_{n\to\infty}\frac{1}{n}\mc{C}_{\mathrm{max}}^\varepsilon(\mc{N}^{\otimes n}).
\end{equation}
\end{result}

\begin{proof}
  The proof follows the same spirit as Lemma~\ref{equivalent}. We show that a simulation protocol can be converted to a (parallel) dilution protocol, and vice versa. For any channel simulation protocol $$\mc{M}(\psi_k,\cdot)\approx\mc{N}(\cdot),$$
  there is a corresponding dilution protocol $$\mc{M}\circ(\mc{G}_k\otimes \mathrm{id}_1)\otimes \mathrm{id}_2\approx\mc{N},$$
  which puts input states to the input of $\mathrm{id}_1$. This shows that $\mc{C}_{\mathrm{sim}}(\mc{N})\geq \mc{C}_{\mathrm{dilute}}(\mc{N})$. Conversely, for any dilution protocol $$\mc{M}_2\circ(\mc{G}_k\otimes \mathrm{id})\circ\mc{M}_1\approx\mc{N},$$
  observe that for any input state $\rho$, the state before $\mc{M}_2$ is $\psi_k\otimes \rho'$ where $\rho'$ is a subsystem of $\mc{M}_1(\rho)$. Since the concatenation of $\mc{M}_1$, partial trace, and $\mc{M}_2$ is also a MIO, this protocol is also a simulation protocol, which means that $\mc{C}_{\mathrm{sim}}(\mc{N})\leq \mc{C}_{\mathrm{dilute}}(\mc{N})$. As the asymptotic channel simulation rate equals to the regularized max-entropy of channel coherence~\cite{Diaz2018usingreusing}, the result is obtained.
\end{proof}

\end{document}